\documentclass{CSML}

\usepackage{lastpage}
\newcommand{\dOi}{14(1:10)2018}
\lmcsheading{}{1--\pageref{LastPage}}{}{}%
{Jun.~12,~2017}{Jan.~23, 2018}{}

\usepackage[hidelinks]{hyperref}

\usepackage{url}
\usepackage{graphicx}
\usepackage{amsmath}
\usepackage{amssymb}
\usepackage{amsthm}
\usepackage{amsxtra,color}
\usepackage{proof}
\usepackage{bussproofs}
\usepackage{float}
\usepackage{soul}
\usepackage{pstricks}
\usepackage{pst-all}

\DeclareSymbolFont{forpolishl}{T1}{cmr}{m}{n}
\DeclareMathSymbol{\mathrmL}{0}{forpolishl}{'212}

\newcommand{\N}{\ensuremath{\mathbb{N}}}

\newcommand{\R}{\ensuremath{\mathbb{R}}}
\newcommand{\pfa}[1]{\text{\footnotesize $#1$}}
\newcommand{\lang}{\ensuremath{\mathcal{L}}}
\newcommand{\fm}{\ensuremath{{\rm Fm}}}
\newcommand{\var}{\ensuremath{{\rm Var}}}
\newcommand{\newvar}{\ensuremath{{\rm Var}^*}}

\newcommand{\bo}{\ensuremath{\Box}}
\newcommand{\di}{\ensuremath{\Diamond}}
\newcommand{\luknot}{\ensuremath{{\sim}}}
\newcommand{\lukimp}{\ensuremath{\supset}}
\newcommand{\fram}[1]{\mathfrak{#1}}

\newcommand{\mdl}[1]{\models_{\lgc{#1}}}
\newcommand{\der}[1]{\vdash_{\lgc{#1}}}
\newcommand{\lgc}[1]{{\mathrm{#1}}}

\newcommand{\f}{\ensuremath{\varphi}}
\newcommand{\p}{\ensuremath{\psi}}
\newcommand{\x}{\ensuremath{\chi}}

\newcommand{\0}{\ensuremath{\overline{0}}}

\newcommand{\De}{\mathrm{\Delta}}
\newcommand{\Ga}{\mathrm{\Gamma}}
\newcommand{\Si}{\mathrm{\Sigma}}
\newcommand{\seq}{\Rightarrow}
\newcommand{\idr}{(\textsc{id})}

\newcommand{\ilr}{(\to\seq)}
\newcommand{\irr}{(\seq\to)}
\newcommand{\salr}{(\&\!\seq)}
\newcommand{\sarr}{(\seq\!\&)}
\newcommand{\nlr}{(\lnot\!\seq)}
\newcommand{\nrr}{(\seq\!\lnot)}
\newcommand{\zrlr}{(\0\!\seq)}
\newcommand{\zrrr}{(\seq\!\0)}
\newcommand{\cutr}{\textup{\sc (cut)}}
\newcommand{\canr}{\textup{\sc (can)}}

\newcommand{\mixr}{\textup{\sc (mix)}}

\newcommand{\seqcontn}{\textup{\sc (sc$_n$)}}
\newcommand{\seqcontk}{\textup{\sc (sc$_k$)}}
\newcommand{\boxknr}[1]{(\bo_{#1})}

\newcommand{\aineq}{\triangleright}
\newcommand{\ex}{(\ensuremath{\rm ex})}

\newcommand{\vecn}[1]{\mathbf{#1}}
\newcommand{\multlr}{(\&\,\aineq)}
\newcommand{\multrr}{(\aineq\,\&)}
\newcommand{\bolrn}{(\bo\,\aineq')}
\newcommand{\borrn}{(\aineq\,\bo')}
\newcommand{\implr}{({\to}\,\aineq)}
\newcommand{\imprr}{(\aineq\,{\to})}
\newcommand{\lorr}{(\aineq\,\lor)}
\newcommand{\lorl}{(\lor\,\aineq)}
\newcommand{\wedr}{(\aineq\,\land)}
\newcommand{\wedl}{(\land\,\aineq)}
\newcommand{\bolr}{(\bo\,\aineq)}
\newcommand{\borr}{(\aineq\,\bo)}
\newcommand{\zerolr}{(\0\,\aineq)}
\newcommand{\zerorr}{(\aineq\,\0)}

\newcommand{\lab}[2]{\ensuremath{(#1)^{#2}}}
\newcommand{\newv}[1]{\lceil #1 \rceil}

\begin{document}

\title[A Real-Valued Modal Logic]{A Real-Valued Modal Logic}
\titlecomment{A precursor to this paper, reporting preliminary results, appeared in the proceedings of AiML 2016~\cite{DMS16}.}

\author[D. Diaconescu]{Denisa Diaconescu\rsuper{a}}
\author[G. Metcalfe]{George Metcalfe\rsuper{b}}	
\author[L. Schn{\"u}riger]{Laura Schn{\"u}riger\rsuper{c}}

\address{\lsuper{a}Faculty of Mathematics and Computer Science, University of Bucharest, Romania}
\thanks{The research of the first author was supported by Sciex grant 13.192 and Romanian National Authority for Scientific Research and Innovation grant, CNCS-UEFISCDI, project number PN-II-RU-TE-2014-4-0730.}
\email{ddiaconescu@fmi.unibuc.ro}

\address{\lsuper{b,c}Mathematical Institute, University of Bern,  Switzerland}	
\email{\{george.metcalfe,laura.schnueriger\}@math.unibe.ch}  
\thanks{The second and third authors acknowledge support from Swiss National Science Foundation grant 200021{\_}146748 and the EU Horizon 2020 research and innovation programme under the Marie Sk{\l}odowska-Curie grant agreement No 689176.}	

\keywords{Many-Valued Logic, Modal Logic, Abelian Logic, \L ukasiewicz Logic, Proof Theory, Tableau Calculus, Sequent Calculus}
\subjclass{F.4.1, I.2.3}


\begin{abstract}
A many-valued modal logic is introduced that combines the usual Kripke frame semantics of the modal logic K with connectives interpreted locally at worlds by lattice and group operations over the real numbers. A labelled tableau system is provided and a  {\sc coNEXPTIME} upper bound obtained for checking validity in the logic. Focussing on the modal-multiplicative fragment, the labelled tableau system is then used to establish completeness for a sequent calculus that admits cut-elimination and an axiom system that extends the multiplicative fragment of Abelian logic.
\end{abstract}


\maketitle

 
\section{Introduction}

Many-valued modal logics combine the frame semantics of classical modal logics with a many-valued semantics at each world. As in the classical case, they may be understood as a compromise between the good computational properties (decidability and lower complexity) of propositional logics and the expressivity of their first-order counterparts, some of which are not even recursively axiomatizable. Such logics have been used to model modal notions such as necessity, belief, and spatio-temporal relations in the presence of multiple degrees of truth, certainty, and possibility, and span fuzzy belief~\cite{HEGG94,GHE03}, fuzzy similarity measures~\cite{GR99}, many-valued tense logics~\cite{HHV95,DG07}, and spatial reasoning with vague predicates~\cite{SDK09}. They also provide a basis for studying fuzzy description logics, which, analogously to the classical case, may be understood as many-valued multi-modal logics (see, e.g.,~\cite{straccia01a,Haj05,KPS13,BBP17}). 

Uniform approaches to many-valued modal logics defined over algebras with a complete lattice reduct are described in~\cite{BEGR11,priest:many}, extending previous work  on modal logics based on finite Heyting algebras~\cite{Fitting91,Fitting92}. In an infinite-valued setting, two core families emerge: ``order-based'' modal logics, including modal extensions of G{\"o}del logics~\cite{CR10,MO11,CR12,CMRR17}, where only the order type of the truth values matters, and ``continuous" modal logics, such as those based on \L ukasiewicz logic~\cite{Haj98,BEGR11,HT13,MM14,MM17,Bar17}, where propositional connectives are interpreted by continuous functions over sets of real numbers (see also~\cite{KP10,KPS13,MS13} for related systems). Such logics are easy to define semantically --- just decide on a suitable set of values and operations --- but not so easy to study. For example, an axiomatization for the G{\"o}del modal logic over many-valued frames is provided in~\cite{CR12}, but no axiomatization is yet known for the G{\"o}del modal logic over standard (Boolean-valued or ``crisp'') frames. Moreover, decidability and complexity problems for these and other order-based modal logics, which typically lack the finite model property, have been solved only recently (see~\cite{CMRR17}).

In this paper we focus on continuous modal logics. Axiomatizations for finite-valued \L ukasiewicz modal logics have been provided in~\cite{HT13}, but the axiom system presented for the infinite-valued  \L ukasiewicz modal logic $\lgc{K(\mathrmL)}$ includes a rule with infinitely many premises. Similarly, only an approximate completeness result (corresponding to including an infinitary rule) is established for the closely related continuous propositional modal logic considered in~\cite{Bar17}. Studying logics that lack a finitary axiom system, and therefore also a suitable algebraic semantics, may be difficult, as may be seen by considering classical modal logic deprived of the theory of Boolean algebras with operators. Note also that, although validity in finite-valued \L ukasiewicz modal logics is PSPACE-complete~\cite{BCE11}, only a  {\sc coNEXPTIME} upper bound is known for the infinite-valued case, as may be deduced from complexity results for \L ukasiewicz description logics (see~\cite{KPS13}).

We address some of these issues here by defining and investigating a many-valued modal logic $\lgc{K(A)}$ with propositional connectives interpreted as the usual lattice and group operations over the real numbers. According to this semantics, the logic may be viewed as a minimal modal extension of Abelian logic, the logic of lattice-ordered abelian groups, introduced independently by Meyer and Slaney as a relevant logic~\cite{mey:ab} and Casari as a comparative logic~\cite{cas:ab}. Some refinements to the usual definition of many-valued modal logics (see, e.g.,~\cite{BEGR11}) are needed to deal with the fact that the real numbers do not form a complete lattice. However, since $\lgc{K(A)}$ enjoys a finite model property, these non-standard features can  be safely ignored for practical purposes. Indeed, the logic $\lgc{K(A)}$ provides a rather simple formalism for reasoning about state transition systems where  linear combinations of real-valued variables are compared among worlds using modal operators.  Since the connectives are interpreted by common arithmetical operations $\min$, $\max$, $+$, and $-$, further connectives interpreted by combinations of these operations (e.g., for many-valued modal or description logics that reason about degrees of truth, certainty, and possibility) can also be defined in this setting. In particular, we show here that the {\L}ukasiewicz modal logic $\lgc{K(\mathrmL)}$ can be interpreted in the logic $\lgc{K(A)}$ extended with a constant. 
 
As our first main technical contribution, we present a sound and complete labelled tableau calculus for $\lgc{K(A)}$ and obtain a {\sc coNEXPTIME} upper bound for checking validity in this logic. The calculus is quite closely related to a labelled tableau calculus for a \L ukasiewicz description logic presented in~\cite{KPS13}, and indeed provides the same upper bound for checking validity. However, an important advantage of defining a logic over lattice-ordered abelian groups is that we are able to explicitly identify and study the modal-multiplicative fragment, solving problems for this fragment that seem at the moment to be quite difficult for the full logic. In particular, we show that the modal-multiplicative fragment of $\lgc{K(A)}$ has an {\sc EXPTIME} upper bound for checking validity and provide a sequent calculus for the fragment that admits cut-elimination. More significantly, we use the  labelled tableau calculus to establish the completeness of an axiom system extending the multiplicative fragment of Abelian logic that, unlike other known axiomatizations for continuous modal logics,  contains only finitary rules.


\section{A Modal Extension of Abelian Logic} \label{s:modalabelianlogic}

In this section we define the real-valued modal logic $\lgc{K(A)}$ semantically as a minimal modal extension of Abelian logic $\lgc{A}$, the logic of lattice-ordered abelian groups. We then show that validity in this logic remains unchanged when the semantics is restricted to the class of finite serial models. Finally, we provide a syntactic embedding of the minimal modal extension $\lgc{K(\mathrmL)}$ of infinite-valued \L ukasiewicz logic into $\lgc{K(A)}$ with an additional constant.

Since we will consider several propositional languages in this paper, let us begin with some quite general definitions. Given a propositional language $\lang$ (also known as an algebraic signature or type) consisting of connectives with fixed arities, let $\fm(\lang)$ denote the set of {\em $\lang$-formulas} $\f,\p,\x,\dots$, defined inductively in the usual way over a countably infinite set $\var$ of (propositional) variables $p,q,r,\dots$. The {\em complexity} of $\f \in \fm(\lang)$ is the number of occurrences of connectives in $\f$, and if $\lang$ contains a unary operation $\bo$, then the {\em modal depth} of $\f$ is  the deepest nesting of the modal connective $\bo$ in $\f$.


\subsection{Abelian Logic}

Let us begin with a brief summary of Abelian logic $\lgc{A}$, introduced independently by Meyer and Slaney in~\cite{mey:ab} as a relevant logic, and Casari in~\cite{cas:ab} as a comparative logic. In both settings, $\lgc{A}$ was defined via axiom systems that are complete with respect to validity in the variety of lattice-ordered abelian groups. However, since this variety is generated by a single algebra defined over the real numbers, we may also use this algebra to  introduce Abelian logic semantically as a many-valued logic.

Consider a language $\lang_\lgc{A}$ with binary connectives $\land$, $\lor$, $\&$, and $\to$, and a constant $\0$, fixing also $\lnot \f := \f \to \0$. We define {\em Abelian logic} $\lgc{A}$ via the logical matrix $\langle \mathbf{R}, \R^+_0 \rangle$ consisting of the algebra $\mathbf{R} = \langle \R,\min,\max,+,-,0 \rangle$ and the set of designated truth values $\R^+_0 = \{r \in \mathbb{R} : r \ge 0\}$. That is, an {\em $\lgc{A}$-valuation} is a map $e \colon \var \to \R$ extended to all $\lang_\lgc{A}$-formulas by
\[
\begin{array}{rcl}
e(\f \land \p) & = & \min(e(\f),e(\p))\\[.025in]
e(\f \lor \p) & = & \max(e(\f),e(\p))\\[.025in]
e(\f \& \p) & = & e(\f) + e(\p)\\[.025in]
e(\f \to \p) & = & e(\p) - e(\f)\\[.025in]
e(\0) & = & 0,
\end{array}
\]
and $\f \in \fm(\lang_\lgc{A})$ is \emph{$\lgc{A}$-valid} if $e(\f) \ge 0$ for each $\lgc{A}$-valuation $e$.

\begin{figure}[tbp] 
\centering
\fbox{
\begin{minipage}{10 cm}
\begin{center}
\[
\begin{array}[t]{rl}
{\rm (B)}  & (\f  \to \p) \to ((\p \to \x) \to (\f \to \x))\\
{\rm (C)}  & 	(\f \to (\p \to \x)) \to (\p \to (\f \to \x))\\
{\rm (\, I \, )}    &	 \f \to \f\\
{\rm (A)}  & 	 ((\f \to \p) \to \p) \to \f\\
(\& 1)      & \f \to (\p \to (\f \& \p))\\
(\& 2)      & (\f \to (\p \to \x)) \to ((\f \& \p) \to \x)\\
(\0\, 1)      & \0 \\
(\0\, 2)      & \f \to (\0 \to \f)\\
(\land 1)   & (\f \land \p) \to \f\\
(\land 2)   & (\f \land \p) \to \p\\
(\land 3)   & ((\f \to \p) \land (\f \to \x)) \to (\f \to (\p \land \x))\\
(\lor 1)    & \f \to (\f \lor \p)\\
(\lor 2)       & \p \to (\f \lor \p)\\
(\lor 3)       & ((\f \to \x) \land (\p \to \x)) \to ((\f \lor \p) \to \x)\\
\end{array}
\]
\[
\infer[(\text{mp})]{\p}{\f & \f \to \p} \qquad  \infer[(\text{adj})]{\f \land \p}{\f & \p}
\]
\end{center}
\caption{The Axiom System $\lgc{HA}$}
\label{f:ha}
\end{minipage}}
\end{figure}

As mentioned above, $\mathbf{R}$ generates the variety of lattice-ordered abelian groups. But also, using methods of abstract algebraic logic, it is easily proved that this variety provides an algebraic semantics for the axiom system $\lgc{HA}$ displayed in Figure~\ref{f:ha}: an axiomatization of multiplicative additive intuitionistic linear logic with just one constant $\0$ extended with the axiom schema ${\rm (A)}\ ((\f \to \p) \to \p) \to \f$. It follows then that $\f \in \fm(\lang_\lgc{A})$ is derivable in $\lgc{HA}$ if and only if $\f$ is $\lgc{A}$-valid.  

The choice of Abelian logic as the basis for the many-valued modal logics studied in this paper is motivated both by its expressivity and the central role of its semantics in ordinary mathematics. The connectives of $\lgc{A}$ are interpreted by the basic arithmetical operations $\min$, $\max$, $+$, $-$, and $0$, from which connectives for other many-valued logics, interpreted via combinations of these operations, can be defined. In particular, there exist syntactic embeddings of  infinite-valued \L ukasiewicz logic into  $\lgc{A}$ that allow results for the latter to be transferred to results about the former. The use of basic arithmetical operations on the real numbers means also that a huge array of results and methods from linear algebra are available for investigating $\lgc{A}$ and its modal expansions. For example, such methods have been used to obtain analytic sequent and hypersequent calculi and co-NP completeness results for Abelian logic and  infinite-valued \L ukasiewicz logic   in~\cite{met:seq} (see also~\cite{BM08,MOG08}).


\subsection{Kripke Semantics}

We define a minimal (crisp) modal extension $\lgc{K(A)}$ of Abelian logic by interpreting formulas locally in the algebra $\mathbf{R}$ over standard Kripke frames. That is, a {\em (crisp) frame} is a pair $\fram{F} = \langle W, R \rangle$, where $W$ is a non-empty set of \emph{worlds} and $R\subseteq W \times W$ is an {\em (crisp) accessibility relation}. As usual, we write $Rxy$ or $Rxy = 1$ to denote $\langle x,y \rangle \in R$ and $Rxy = 0$ to denote $\langle x,y \rangle \not\in R$. For any $x\in W$, we let $R[x] = \{y \in W : Rxy\}$. Modal formulas are defined over the  language $\lang_{\lgc{A}}^\bo$ extending $\lang_\lgc{A}$ with an additional unary ``box'' connective $\bo$, where the dual ``diamond'' connective is defined as $\di \f := \lnot \bo \lnot \f$. 

There exists a very general method for defining crisp modal logics over algebras with a complete lattice reduct (see in particular~\cite{BEGR11}), where the $\bo$ and $\di$ connectives are interpreted as infima and suprema of values of formulas at accessible worlds. However, since the real numbers do not form a complete lattice --- they lack a top and bottom element --- we make here a couple of minor adjustments to this method. First, we adopt the useful convention that $\bigwedge_\R \emptyset = \bigvee_\R \emptyset = 0$, and second, we restrict valuations of variables in a particular model to a fixed interval. Both these choices will be justified to some extent by Lemma~\ref{l:fmp} below.

A \emph{$\lgc{K(A)}$-model} $\fram{M} = \langle W, R, V \rangle$ consists of a frame $\langle W, R \rangle$ together with a {\em valuation} map $V \colon \var \times W \to [-r,r]$ for some $r \in \R^+_0$ that is extended to $V \colon \fm(\lang_{\lgc{A}}^\bo) \times W \to \R$ by 
\[
\begin{array}{rcl}
V(\f \land \p, x) & = & \min(V(\f,x),V(\p,x))\\[.025in]
V(\f \lor \p, x) & = & \max(V(\f,x), V(\p,x))\\[.025in]
V(\f \& \p, x) & = & V(\f,x) + V(\p,x)\\[.025in]
V(\f \to \p, x) & = & V(\p,x) - V(\f,x)\\[.025in]
V(\0, x) & = & 0\\[.025in]
V(\bo \f, x) & = &  \bigwedge_\R \{V(\f, y) : Rxy \}.
\end{array}
\]
By calculation, we obtain also
\[
\begin{array}{rcl}
V(\lnot \f, x) & = & - V(\f,x)\\[.025in]
V(\di \f, x) 	& = & \bigvee_\R \{V(\f, y) : Rxy \}.
\end{array}
\]
An $\lang_{\lgc{A}}^\bo$-formula $\f$ is \emph{valid} in a $\lgc{K(A)}$-model $\fram{M} = \langle W, R,V \rangle$ if $V(\f, x) \ge 0$ for all $x \in W$. If $\f$ is valid in all $\lgc{K(A)}$-models, then $\f$ is \emph{$\lgc{K(A)}$-valid}, written $\mdl{\lgc{K(A)}} \f$.

The convention that $\bigwedge_\R \emptyset = \bigvee_\R \emptyset = 0$ is rather counter-intuitive. This can be avoided, however, by restricting to {\em serial} frames: that is, frames $\fram{F} = \langle W, R \rangle$ such that for all $x \in W$, there exists $y \in W$ such that $Rxy$. With this restriction, $\bigwedge_\R \emptyset$ and $\bigvee_\R \emptyset$ may simply be left undefined. Similarly, restricting the codomain of a valuation to a bounded subset of $\R$ can be avoided by considering only finite models. Surprisingly perhaps, considering only finite serial models does not affect the valid formulas of the logic.

\begin{lem}\label{l:fmp}
$\mdl{\lgc{K(A)}} \f$ if and only if $\f$ is valid in all finite serial $\lgc{K(A)}$-models.
\end{lem}
\begin{proof}
The left-to-right direction is immediate. For the opposite direction, we note first that if $\f$ is not valid in a $\lgc{K(A)}$-model $\fram{M} = \langle W, R,V \rangle$, then it will not be valid in the serial $\lgc{K(A)}$-model $\fram{M}' = \langle W \cup W', R \cup R', V' \rangle$ where $W'$ is a set of new distinct worlds $x'$ for each $x \in W$ satisfying $R[x] = \emptyset$, $R' = \{\langle x,x'\rangle, \langle x',x'\rangle : x \in W \text{ and } R[x] = \emptyset\}$, and $V'$ extends $V$ with $V(p,x') = 0$ for all $p \in \var$ and $x' \in W'$. Clearly, if $\fram{M}$ is finite, then $\fram{M}'$ is also finite.

It now suffices to prove the following: for any $\lgc{K(A)}$-model $\fram{M} = \langle W, R, V \rangle$, $x \in W$, finite set of formulas $S$, and $\varepsilon > 0$, there exists a finite $\lgc{K(A)}$-model $\fram{M}' = \langle W', R', V' \rangle$ with $x \in W'$ such that $|V(\f,x) -V'(\f,x)| < \varepsilon$ for all $\f \in S$. We proceed by induction on the sum of the complexities of the formulas in $S$. 

For the base case, $S$ contains only variables and $\0$, and we let $\fram{M}' = \langle W', R', V' \rangle$ with $W' = \{x\}$, $R' = \emptyset$, and $V'(p,x) = V(p,x)$ for each $p \in \var$. For the inductive step, suppose first that $S = S' \cup \{\p_1 \to \p_2\}$. Then we can apply the induction hypothesis with $\fram{M}$, $x \in W$, $S'' = S' \cup \{\p_1,\p_2\}$, and $\frac{\varepsilon}{2} > 0$ to obtain a finite $\lgc{K(A)}$-model $\fram{M}'= \langle W', R', V' \rangle$ with $x \in W'$ such that $|V(\f,x) -V'(\f,x)| < \frac{\varepsilon}{2}$ for all $\f \in S''$. It suffices then to observe that $|V(\p_1 \to \p_2,x) -V'(\p_1 \to \p_2,x)| = |V(\p_2,x) - V(\p_1,x) - V'(\p_2,x) + V'(\p_1,x)| \le |V(\p_2,x) - V'(\p_2,x)| + |V(\p_1,x)  - V'(\p_1,x)| < \frac{\varepsilon}{2} + \frac{\varepsilon}{2} = \varepsilon$. The cases where $S$ contains $\p_1 \& \p_2$, $\p_1 \land \p_2$, or $\p_1 \lor \p_2$ are very similar.

Finally, suppose that $S$ consists of  variables and boxed formulas $\bo \p_1,\ldots,\bo \p_n$ ($n \ge 1$). Then for $i \in \{1,\ldots,n\}$, there exists $y_i \in W$ such that $Rxy_i$ and $|V(\bo \p_i,x) - V(\p_i,y_i)| < \frac{\varepsilon}{2}$. We apply the induction hypothesis to each submodel $\fram{M}_i$ of $\fram{M}$ generated by $y_i$ (i.e., the restriction of $\fram{M}$ to the smallest subset of $W$ containing  $y_i$ and closed under $R$) with $S' = (S \setminus \{\bo \p_1,\ldots,\bo \p_n\}) \cup \{\p_1,\ldots,\p_n\}$, $y_i \in W_i$, and $\frac{\varepsilon}{2} > 0$ to obtain a finite $\lgc{K(A)}$-model $\fram{M}'_i= \langle W'_i, R'_i, V'_i \rangle$ and $y_i \in W'_i$ such that $|V(\f,y_i) -V'(\f,y_i)| < \frac{\varepsilon}{2}$ for all $\f \in S'$. By renaming worlds, we may assume that these models are disjoint and do not include $x$. Now let $\fram{M}'= \langle W', R', V' \rangle$ be the finite  $\lgc{K(A)}$-model with $W' = \{x\} \cup W'_1 \cup \ldots \cup W'_n$ such that for $u,v \in W'$ and $p \in \var$,
\[
\begin{array}{rclcrcl}
R'uv & = &
\begin{cases}
R'_iuv & \text{if } u,v \in W'_i\\
1		 & \text{if } u=x, \ v \in \{y_1,\ldots,y_n\}\\
0		 & \text{otherwise}
\end{cases}
& \quad\text{and}\quad &
V'(p,u) & = &
\begin{cases}
V'_i(p,u) & \text{if } u \in W'_i\\
V(p,x)	& \text{if } u=x.
\end{cases}
\end{array}
\]
Clearly $V'(p,x) = V(p,x)$ for each variable $p \in S$. Moreover, $|V(\bo \p_i,x) - V(\p_i,y_i)| < \frac{\varepsilon}{2}$ and  $|V(\p_i,y_j) - V'(\p_i,y_j)| <  \frac{\varepsilon}{2}$ for $i,j \in \{1,\ldots,n\}$, so $|V(\bo \p_i,x) - V'(\bo \p_i,x)| < \varepsilon$.
\end{proof}

As remarked above, the preceding lemma provides some justification both for assuming in the definition of the semantics of $\lgc{K(A)}$ that $\bigwedge_\R \emptyset = \bigvee_\R \emptyset = 0$, and for restricting valuations of variables in a particular $\lgc{K(A)}$-model to a fixed interval. It shows that for determining the valid formulas of $\lgc{K(A)}$, we need only consider finite serial $\lgc{K(A)}$-models. For such frames we can leave $\bigwedge_\R \emptyset$ and $\bigvee_\R \emptyset$ undefined; we can also make use of a standard (unrestricted) valuation map $V \colon \var \times W \to \mathbb{R}$, since  only finite infima and suprema are needed for calculating values of formulas. In principle then, we could define the logic $\lgc{K(A)}$ without extra assumptions by considering only finite serial $\lgc{K(A)}$-models. We prefer here, however, to give a more general semantics and to discover this finite model property as a fact about the logic rather than building it into the definition.

\begin{exa}
Any $\lgc{K(A)}$-model can be viewed as a state transition system, where each state is labelled with a vector of real numbers that represents the values of the variables in $\var$  at that state. Such a transition system may be used to represent choices for various players in a game together with points (or other resources) accumulated by the players during that game. Consider for example  the $\lgc{K(A)}$-model $\fram{M} = \langle W, R, V \rangle$ depicted below, where the vectors {\tiny $\begin{pmatrix} p \\ q \end{pmatrix}$} represent the values of $p$ and $q$, respectively, at each state.

 
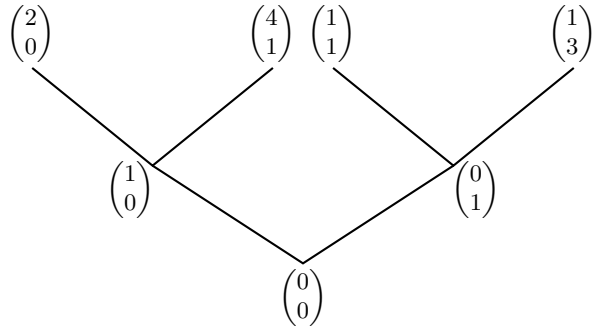
\begin{center}
 \begin{pspicture}(-5,-1)(5,3.75)
 \psline[](-2,1.3)(0,0)(2,1.3)
 \psline[](-3.6,2.6)(-2,1.3)(-0.4,2.6)
 \psline[](0.4,2.6)(2,1.3)(3.6,2.6)
 \uput[-90](0,0.1){\footnotesize $\begin{pmatrix} 0 \\ 0 \end{pmatrix}$}
 \uput[-150](-1.8,1.3){\footnotesize $\begin{pmatrix} 1 \\ 0 \end{pmatrix}$}
 \uput[-30](1.8,1.3){\footnotesize $\begin{pmatrix} 0 \\ 1 \end{pmatrix}$}
 \uput[40](-4.1,2.6){\footnotesize $\begin{pmatrix} 2 \\ 0 \end{pmatrix}$}
 \uput[40](-0.9,2.6){\footnotesize $\begin{pmatrix} 4 \\ 1 \end{pmatrix}$}
 \uput[40](-0.1,2.6){\footnotesize $\begin{pmatrix} 1 \\ 1 \end{pmatrix}$}
 \uput[40](3.1,2.6){\footnotesize $\begin{pmatrix} 1 \\ 3 \end{pmatrix}$}
 \end{pspicture}
 \end{center}

\noindent 
We can use the model $\fram{M}$ to define various two-player games, where in the first round, starting at the root, Player~$P$ chooses one of several (in this case, two) options, and in the second round, Player~$Q$ also chooses one of several (in this case, also two) options. The points assigned to Player~$P$ and Player~$Q$ at each state are the values of $p$ and $q$, respectively. Let also call the values of $p-q$ and $q-p$ at a state, the {\em scores} for~$P$ and~$Q$, respectively. We assume in all these games that the players have complete knowledge of both $\fram{M}$ and their opponent's goals.

Let us consider some different ways of concluding games based on $\fram{M}$. Suppose that in Game 1 each player's goal is to maximize her final score. Player $P$'s maximal payoff is then the value at the root of the formula $\di \bo (q \to p)$, which is $2$. If Player~$P$'s goal in Game~2 is to maximize her final number of points, and Player~$Q$ aims to minimize this number, then the required formula is $\di \bo p$, which at the root also takes value $2$. Reversing the roles for Game $3$, we obtain the formula $\bo \di q$, which takes value $1$ at the root. More complicated goals can also be modelled. For example, if both players aim to maximize the sum of their scores accumulated during the two rounds, then Player~$P$'s maximal payoff is the value of the formula $\di ((q \to p) \& \bo (q \to p))$ at the root, namely $3$.

Formulas can also be used to express general relationships between games. For example, the model $\fram{M}$ shows that
\[
\not \mdl{\lgc{K(A)}} \di \bo (q \to p) \to (\bo \di q \to \di \bo p).
\]
That is, Player~$P$'s maximal payoff in Game 1 exceeds her maximal payoff in Game 2 minus her maximal payoff in Game 3. On the other hand, it can be shown (e.g., using one of the calculi introduced below) that
\[
\mdl{\lgc{K(A)}} \di \bo (q \to p) \to (\bo \bo q \to \di \bo p).
\]
This means that if the goals of Games 1 and 2 are adopted with respect to an arbitrary $\lgc{K(A)}$-model $\fram{M}'$ based on a finite rooted tree with branches of length $2$, then Player~$P$'s maximal payoff in Game 1 for $\fram{M}'$  is always less than or equal to her maximal payoff in Game 2 for $\fram{M}'$ minus the minimum number of points for Player~$Q$ in all final states of $\fram{M}'$. 
\end{exa}


\subsection{\L ukasiewicz Modal Logic} \label{ss:lukmodal}

Let us briefly recall the semantics of the \L ukasiewicz modal logic $\lgc{K(\mathrmL)}$ studied by Hansoul and Teheux in~\cite{HT13}. For convenience, we make use of a language $\lang_{\mathrmL}^\bo$ with the binary connective $\lukimp$ and unary connectives $\luknot$ and $\bo$, where further connectives are defined as $\f \oplus \p := \luknot \f \lukimp \p$, $\f \odot \p :=  \luknot (\luknot \f \oplus \luknot \p)$, $\f \lor	 \p  := (\f \lukimp \p) \lukimp \p$, $\f \land \p := \luknot (\luknot \f \lor \luknot \p)$, and $\di \f 	:=  \luknot \bo \luknot \f$.

A \emph{$\lgc{K(\mathrmL)}$-model} $\fram{M} = \langle W, R, V \rangle$ consists of a frame $\langle W, R \rangle$ and a valuation map $V \colon \var \times W \to [0,1]$ that is extended to $V \colon \fm(\lang_{\mathrmL}^\bo) \times W \to [0,1]$ by 
\[
\begin{array}{rcl}
V(\luknot \f, x) & = & 1 - V(\f,x)\\[.025in]
V(\f \lukimp \p, x) & = & \min(1, 1 - V(\f,x) + V(\p,x))\\[.025in]
V(\bo \f, x) & = &  \bigwedge_{[0,1]} \{V(\f, y) : Rxy \}.
\end{array}
\]
An $\lang_{\mathrmL}^\bo$-formula $\f$ is \emph{valid} in a $\lgc{K(\mathrmL)}$-model $\fram{M} = \langle W, R,V \rangle$ if $V(\f,x) = 1$ for all $x \in W$. If $\f$ is valid in all $\lgc{K(\mathrmL)}$-models, then $\f$ is \emph{$\lgc{K(\mathrmL)}$-valid}, written $\mdl{\lgc{K(\mathrmL)}} \f$.
 
An axiom system for $\lgc{K(\mathrmL)}$ is presented in~\cite{HT13} as an extension of an axiomatization of infinite-valued \L ukasiewicz logic with the modal axioms and rules
\[
\begin{array}{c}
\bo (\f \lukimp \p) \lukimp (\bo \f \lukimp \bo \p)\\[.025in]
\bo(\f \oplus \f) \lukimp (\bo \f \oplus \bo \f)\\[.025in]
\bo(\f \odot \f) \lukimp (\bo \f \odot \bo \f)\\[.075in]
\infer{\bo \f}{\f} 
\end{array}
\]
and the following rule with infinitely many premises
\[
 \infer{\f}{\f \oplus \f & \f \oplus (\f \odot \f) & \f \oplus (\f \odot \f \odot \f) & \ldots}
 \]
It is proved that an $\lang_{\mathrmL}^\bo$-formula $\f$ is derivable in this system if and only if $\mdl{\lgc{K(\mathrmL)}} \f$.\footnote{In fact, the authors of~\cite{HT13} prove a more general {\em strong completeness} result: an $\lang_{\mathrmL}^\bo$-formula $\f$ is derivable from a (possibly infinite) set of $\lang_{\mathrmL}^\bo$-formulas $\Si$ in the system if and only if for every $\lgc{K(\mathrmL)}$-model $\langle W, R,V \rangle$ and $x \in W$, whenever $V(\p,x) = 1$ for all $\p \in \Si$, also $V(\f,x) = 1$. Note that an infinitary rule is needed to obtain  a strong completeness theorem even for propositional {\L}ukasiewicz logic and Abelian logic. However, in this paper we establish only (weak) completeness results.}

This raises an intriguing question. Is there an elegant axiomatization containing only finitary rules, obtained perhaps by removing the infinitary rule above? Our first step towards addressing this issue will be to view \L ukasiewicz modal logic $\lgc{K(\mathrmL)}$ as a fragment of a modest extension of the Abelian modal logic $\lgc{K(A)}$. Let $\lang_{\lgc{A^c}}^\bo$ be the language $\lang_{\lgc{A}}^\bo$ extended with an extra constant $c$. A $\lgc{K(A^c)}$-model $\fram{M} = \langle W, R, V, c^\fram{M} \rangle$ consists of a $\lgc{K( A)}$-model $\langle W, R,V \rangle$ and an element $c^\fram{M} \in \R$, where valuations are extended as before, using the additional clause $V(c,x) =  c^\fram{M}$ for each $x \in W$. 

Let us fix $\bot := c \land \lnot c$ and define the following mapping from $\fm(\lang_{\mathrmL}^\bo)$ to $\fm(\lang_{\lgc{A^c}}^\bo)$:
\[
\begin{array}{rcl}
p^* 			& = & (p \land \0) \lor \bot \ \text{ for each} \ p \in \var\\[.025in]
(\luknot \f)^*  		& = & \f^* \to \bot\\[.025in]
(\f \lukimp \p)^*	& = & (\f^* \to \p^*) \land \0\\[.025in]
(\bo \f)^*		& = & \bo \f^*.
\end{array}
\]
We show that this mapping preserves validity between $\lgc{K(\mathrmL)}$ and $\lgc{K(A^c)}$ by identifying the value taken by an $\lang_{\mathrmL}^\bo$-formula in $[0,1]$ with the value taken by the corresponding $\lang_{\lgc{A^c}}^\bo$-formula in the interval $[-|c|,0]$.

\begin{prop}
Let  $\f \in \fm(\lang_{\mathrmL}^\bo)$. Then $\mdl{\lgc{K(\mathrmL)}} \f$ if and only if $\mdl{\lgc{K(A^c)}} \f^*$.
\end{prop}
\begin{proof}
Suppose first that $\f$ is not valid in a $\lgc{K(\mathrmL)}$-model $\fram{M} = \langle W, R, V \rangle$. So $V(\f,x) < 1$ for some $x \in W$. We consider the $\lgc{K(A^c)}$-model $\fram{M}' = \langle W, R, V', c^\fram{M} \rangle$ where $V'(p,x) = V(p,x) - 1$ for any $p \in \var$ and $x \in W$, and $c^\fram{M} = -1$. It suffices to prove that $V'(\p^*,x) = V(\p,x) - 1$ for any $\p \in \fm(\lang_{\mathrmL}^\bo)$, since then $V'(\f^*,x) = V(\f,x) - 1 < 0$ and $\not \mdl{\lgc{K(A^c)}} \f^*$. We proceed by induction on the complexity of $\p$. The base case follows by definition and for the inductive step for the propositional connectives, we just notice that, using the induction hypothesis,
\[
\begin{array}{rcl}
V'((\p_1 \lukimp \p_2)^*,x) & = & V'((\p_1^* \to \p_2^*) \land \0,x)\\[.025in]
& = & \min(V'(\p_2^*,x)-V'(\p_1^*,x),0)\\[.025in]
& = & \min((V(\p_2,x) - 1) -(V(\p_1,x) - 1),0)\\[.025in]
& = & \min(V(\p_2,x) - V(\p_1,x),0)\\[.025in]
& = & \min(1 - V(\p_1,x) + V(\p_2,x),1) - 1 \\[.025in]
& = & V(\p_1  \lukimp \p_2,x) - 1,
\end{array}
\]
the case where $\p$ is $\luknot \p_1$ being very similar. For the modal case, we obtain
\[
\begin{array}{rcl}
V'((\bo \p_1)^*,x) & = & V'(\bo \p_1^*,x)\\[.025in]
& = & \bigwedge_{\R} \{V'(\p_1^*,y) : Rxy\}\\[.025in]
& = & \bigwedge_{\R} \{V(\p_1,y) - 1 : Rxy\}\\[.025in]
& = & \bigwedge_{[0,1]} \{V(\p_1,y) : Rxy\} - 1\\[.025in]
& = & V(\bo \p_1,x) - 1,
\end{array}
\]
noting that in the case where $R[x] = \emptyset$, we obtain $\bigwedge_{\R} \{V'(\p_1^*,y) : Rxy\} = 0 = 1 - 1 = \bigwedge_{[0,1]} \{V(\p_1,y) : Rxy\} - 1$ as required.

Suppose now conversely that $\f^*$ is not valid in a $\lgc{K(A^c)}$-model $\fram{M} = \langle W, R, V, c^\fram{M} \rangle$. That is, $V(\f^*,x) < 0$ for some $x \in W$. Observe first that if $c^\fram{M} = 0$, then, by a simple induction on the complexity of $\f$, we obtain  $V(\f^*,x) = 0$ for all $x \in W$, a contradiction. Hence $c^\fram{M} \neq 0$. Moreover,  by scaling (dividing $V(p,x)$ by $c^\fram{M}$ for each $p \in \var$ and $x \in W$), we may assume that $V(\bot,x) = -1$ for all $x \in W$. We consider the $\lgc{K(\mathrmL)}$-model $\fram{M}' = \langle W, R, V' \rangle$ where $V'(p,x) = \max(\min(V(p,x) + 1,1),0)$. It then suffices to prove that $V'(\p,x) = V(\p^*,x) + 1$  for any $\p \in \fm(\lang_{\mathrmL}^\bo)$, proceeding by induction on the complexity of $\p$. 
\end{proof}

Note that the addition of a constant $c$ to $\lgc{K(A)}$ does not affect the fact that validity in the logic is equivalent to validity in finite models. It does, however, introduce a difference between the logic $\lgc{K(A^c)}$ and the same logic restricted to serial models. Clearly, the formula $c \to \bo c$ is valid in all serial models, but not in all models. 


\section{A Labelled Tableau Calculus}


In this section we introduce a labelled tableau calculus for checking $\lgc{K(A)}$-validity that is based very closely on the Kripke semantics described above. We use the calculus here to show that the problem of checking $\lgc{K(A)}$-validity is in the complexity class {\sc coNEXPTIME}. In Section~\ref{s:fragment}, we will also use (a fragment of) the calculus to establish the completeness of an axiom system and a sequent calculus admitting cut-elimination for the modal-multiplicative fragment of $\lgc{K(A)}$.

\subsection{The Calculus}

Our labelled tableau calculus $\lgc{LK(A)}$ proves that an $\lang_{\lgc{A}}^\bo$-formula $\f$ is valid by showing that the assumption that $\f$ takes a value less than $0$ in some world $w_1$ leads to a contradiction. Informally, we  build a tableau for $\f$ as follows. First we decompose the propositional structure of $\f$ to obtain inequations between sums of formulas labelled with the world $w_1$. We then use box formulas occurring on the right of these inequations to generate new worlds accessible to $w_1$ and further inequations between sums of formulas labelled with $w_1$ and these accessible worlds.  Box formulas on the left are decomposed by considering accessible worlds to $w_1$ and generating new inequations for those worlds. The process is then repeated with the new inequations and worlds appearing on the tableau. The formula $\f$ will be valid if the generated set of inequations (suitably interpreted) on each branch of the tableau is unsatisfiable over the real numbers.

By a {\em labelled formula} we mean an ordered pair consisting of an $\lang_{\lgc{A}}^\bo$-formula $\f$ and a natural number $k$, written $\lab{\f}{k}$. Given a multiset of $\lang_{\lgc{A}}^\bo$-formulas $\Ga = [\f_1,\ldots,\f_n]$ (denoting the empty multiset by $[]$) and $k_1,\ldots,k_n \in \N$, we let $\lab{\Ga}{\vecn{k}}$ denote the multiset of labelled formulas $[\lab{\f_1}{k_1},\ldots,\lab{\f_n}{k_n}]$. 

Tableaux are constructed from {\em (tableau) nodes} of two types:\medskip

\begin{enumerate}[label=(\arabic*)]

\item
{\em labelled inequations} of the form $\lab{\Ga}{\vecn{k}} \aineq \lab{\De}{\vecn{l}}$ where $\aineq \in \{>,\geq\}$ and $\lab{\Ga}{\vecn{k}},\lab{\De}{\vecn{l}}$ are finite multisets of labelled formulas;\medskip

\item
{\em relations} of the form $rij$ where $i,j \in \N$.\medskip

\end{enumerate}

\noindent
An  $\lgc{LK(A)}$-{\em tableau} is a finite tree of nodes generated according to the inference rules of the system presented in Figure~\ref{f:lkz}. That is, if nodes above the line in an instance of a rule occur on the same branch $B$, then $B$ can be extended with the nodes below the line. For convenience, we often write branches as (numbered) lists, noting for future reference that tableaux for formulas in the modal-multiplicative fragment (i.e., not containing $\land$ or $\lor$) consist of just one branch.

\newcommand{\premiseszerol}{
\lab{\Ga}{\vecn{k}}, \lab{\0}{i} \aineq \lab{\De}{\vecn{l}}
}
\newcommand{\conclusionszerol}{
\lab{\Ga}{\vecn{k}} \aineq  \lab{\De}{\vecn{l}}
}
\newcommand{\premiseszeror}{
\lab{\Ga}{\vecn{k}} \aineq \lab{\0}{i}, \lab{\De}{\vecn{l}}
}
\newcommand{\conclusionszeror}{
\lab{\Ga}{\vecn{k}} \aineq \lab{\De}{\vecn{l}}
}
\newcommand{\premisesmultl}{
\lab{\Ga}{\vecn{k}}, \lab{\f\&\p}{i} \aineq \lab{\De}{\vecn{l}}
}
\newcommand{\conclusionsmultl}{
\lab{\Ga}{\vecn{k}}, \lab{\f}{i}, \lab{\p}{i} \aineq  \lab{\De}{\vecn{l}}
}
\newcommand{\premisesmultr}{
\lab{\Ga}{\vecn{k}} \aineq \lab{\f \& \p}{i}, \lab{\De}{\vecn{l}}
}
\newcommand{\conclusionsmultr}{
\lab{\Ga}{\vecn{k}} \aineq \lab{\f}{i}, \lab{\p}{i}, \lab{\De}{\vecn{l}}
}
\newcommand{\premisesbor}{
\lab{\Ga}{\vecn{k}} \aineq \lab{\bo \f}{i}, \lab{\De}{\vecn{l}}
}
\newcommand{\conclusionsbor}{
\begin{array}{c}
\lab{\bo \f}{i} \ge \lab{\f}{j} \\
rij 
\end{array}}
\newcommand{\premisesbol}{
\begin{array}{c}
rij\\
\lab{\Ga}{\vecn{k}}, \lab{\bo \f}{i} \aineq \lab{\De}{\vecn{l}}
\end{array}}
\newcommand{\conclusionsbol}{
\begin{array}{c}
\lab{\f}{j} \ge \lab{\bo \f}{i}\\
\end{array}}
\newcommand{\premisesimpl}{
\lab{\Ga}{\vecn{k}}, \lab{\f\to\p}{i} \aineq \lab{\De}{\vecn{l}}
}
\newcommand{\conclusionsimpl}{
\lab{\Ga}{\vecn{k}}, \lab{\p}{i} \aineq \lab{\f}{i}, \lab{\De}{\vecn{l}}
}
\newcommand{\premisesimpr}{
\lab{\Ga}{\vecn{k}} \aineq \lab{\f \to \p}{i}, \lab{\De}{\vecn{l}}
}
\newcommand{\conclusionsimpr}{
\lab{\Ga}{\vecn{k}}, \lab{\f}{i} \aineq \lab{\p}{i}, \lab{\De}{\vecn{l}}
}
\newcommand{\premisesveel}{
\lab{\Ga}{\vecn{k}}, \lab{\f \lor \p}{i} \aineq  \lab{\De}{\vecn{l}}
}
\newcommand{\conclusionsveel}{
\lab{\Ga}{\vecn{k}}, \lab{\f}{i} \aineq \lab{\De}{\vecn{l}} \quad  \lab{\Ga}{\vecn{k}}, \lab{\p}{i} \aineq \lab{\De}{\vecn{l}}
}
\newcommand{\premisesveer}{
\lab{\Ga}{\vecn{k}} \aineq \lab{\f \lor \p}{i}, \lab{\De}{\vecn{l}}
}
\newcommand{\conclusionsveer}{
\begin{array}{c}
\lab{\Ga}{\vecn{k}} \aineq \lab{\f}{i}, \lab{\De}{\vecn{l}} \\ 
\lab{\Ga}{\vecn{k}} \aineq \lab{\p}{i}, \lab{\De}{\vecn{l}}
\end{array}
}
\newcommand{\premiseswedl}{
\lab{\Ga}{\vecn{k}}, \lab{\f \land \p}{i} \aineq  \lab{\De}{\vecn{l}}
}
\newcommand{\conclusionswedl}{
\begin{array}{c}
\lab{\Ga}{\vecn{k}}, \lab{\f}{i} \aineq  \lab{\De}{\vecn{l}} \\ 
\lab{\Ga}{\vecn{k}}, \lab{\p}{i} \aineq  \lab{\De}{\vecn{l}}
\end{array}
}
\newcommand{\premiseswedr}{
\lab{\Ga}{\vecn{k}}  \aineq \lab{\f \land \p}{i}, \lab{\De}{\vecn{l}}
}
\newcommand{\conclusionswedr}{
\lab{\Ga}{\vecn{k}}  \aineq \lab{\f}{i}, \lab{\De}{\vecn{l}} \quad \lab{\Ga}{\vecn{k}} \aineq \lab{\p}{i}, \lab{\De}{\vecn{l}}
}

\begin{figure}[tbp] 
\centering
\fbox{
\begin{minipage}{14.5 cm}
\[
\begin{array}{cc}
\infer[\zerolr]{\conclusionszerol}{\premiseszerol} 
&
\infer[\zerorr]{\conclusionszeror}{\premiseszeror}  \qquad \qquad \qquad \\[.2in]
\infer[\multlr]{\conclusionsmultl}{\premisesmultl} 
&
\infer[\multrr]{\conclusionsmultr}{\premisesmultr}  \qquad \qquad \qquad \\[.2in]
\infer[\implr]{\conclusionsimpl}{\premisesimpl} 
&
\infer[\imprr]{\conclusionsimpr}{\premisesimpr}  \qquad \qquad \qquad \\[.2in]
\infer[\wedl]{\conclusionswedl}{\premiseswedl} 
&
\infer[\wedr]{\conclusionswedr}{\premiseswedr}  \qquad \qquad \qquad  \\[.3in]
\infer[\lorl]{\conclusionsveel}{\premisesveel} 
&
\infer[\lorr]{\conclusionsveer}{\premisesveer}  \qquad \qquad \qquad  \\[.2in]
\infer[\bolr]{\conclusionsbol}{\premisesbol}
& 
\infer[\borr]{\conclusionsbor}{\premisesbor}
\hspace{-.5in}j \in \N \text{ new} \hspace{1.7cm}
\end{array}
\]
\[
\infer[\ex]{rkj}{rik} \ j \in \N \text{ new}
\]
\caption{The labelled tableau calculus $\lgc{LK(A)}$} \label{f:lkz}
\end{minipage}}
\end{figure}

 Observe that the rules for $\0$, $\&$, $\to$, $\land$, and $\lor$ decompose formulas on the left and right of inequations, using the same label for added subformulas, while the rules for $\bo$ introduce inequations between a boxed formula $\bo \f$ labelled with $i$, and $\f$ labelled with a different $j$. Let us note also that the premises $\lab{\Ga}{\vecn{k}}, \lab{\bo \f}{i} \aineq \lab{\De}{\vecn{l}}$ and $\lab{\Ga}{\vecn{k}} \aineq \lab{\bo \f}{i}, \lab{\De}{\vecn{l}}$ of $\bolr$ and $\borr$, respectively, are not, strictly speaking, necessary for either the soundness or the completeness of the calculus. However, they restrict the decomposition of boxed formulas to those occurring as subformulas of the initial formula, thereby ensuring a subformula property for the calculus.

Let ${\rm LVar}$ be the set of all formulas of the form $\lab{p}{i}$ and $\lab{\bo \f}{i}$ for $p \in \var$, $\f \in \fm(\lang_{\lgc{A}}^\bo)$, and $i \in \N$, considered as a set of variables. Given an $\lgc{LK(A)}$-tableau $T$ and a branch $B$ of $T$, the {\em system of inequations $S$ associated to $B$} consists of all labelled inequations occurring on $B$ that contain only formulas from ${\rm LVar}$. Each labelled inequation in $S$ is interpreted here as an inequation between formal sums (where addition is over the multisets occurring in the labelled inequation and the empty multiset is $0$) of variables from ${\rm LVar}$. We call the branch $B$ {\em open} if the set of inequations $S$ associated to $B$ is consistent over $\R$, and {\em closed} otherwise. The tableau $T$ is called {\em closed} if all of its branches are closed, and {\em open} if it has at least one open branch.

A  {\em tableau for an $\lang_{\lgc{A}}^\bo$-formula} $\f$ is an $\lgc{LK(A)}$-tableau with root node $[] > [\lab{\f}{1}]$ and covering node $r12$. We say that $\f$ is {\em $\lgc{LK(A)}$-derivable}, written $\der{\lgc{LK(A)}} \f$, if there exists a closed tableau for $\f$.

\begin{exa}
The seriality axiom $\bo p \to \di p$ is $\lgc{LK(A)}$-derivable using the tableau\smallskip
\begin{center}
\begin{tabular}{rl}
{\footnotesize$1:$}	&	$[] > \lab{\bo p \to (\bo(p \to \0) \to \0)}{1}$\\
{\footnotesize$2:$}	&	$r12$\\
{\footnotesize$3:$}	&	$\lab{\bo p}{1} > \lab{\bo(p \to \0) \to \0}{1}$\\
{\footnotesize$4:$}	&	$\lab{\bo p}{1}, \lab{\bo(p\to\0)}{1} > \lab{\0}{1}$\\
{\footnotesize$5:$}	&	$\lab{\bo p}{1}, \lab{\bo(p\to\0)}{1} > []$\\
{\footnotesize$6:$}	&	$\lab{p}{2} \geq \lab{\bo p}{1}$ \\
{\footnotesize$7:$}	&	$\lab{p \to \0}{2} \geq \lab{\bo(p\to \0)}{1}$ \\
{\footnotesize$8:$}	&	$\lab{\0}{2} \geq \lab{p}{2}, \lab{\bo(p\to \0)}{1}$ \\
{\footnotesize$9:$}	&	$[] \geq \lab{p}{2}, \lab{\bo(p\to \0)}{1}$ \\
\end{tabular}
\end{center}\smallskip
which generates a (single) inconsistent system of inequations over $\R$
\[
\{
  x + y > 0, \
  z \geq x, \
  0 \geq z + y
\}
\]
where  $x$, $y$, and $z$ stand for $\lab{\bo p}{1}$, $\lab{\bo(p \to \0)}{1}$, and $\lab{p}{2}$, respectively.
\end{exa}

The calculus $\lgc{LK(A)}$ can also be used to prove that an $\lang_{\lgc{A}}^\bo$-formula is {\em not} $\lgc{K(A)}$-valid; indeed a concrete counter-model for such a formula can be constructed from an open branch of a tableau where, taking care to avoid loops, the rules have been applied exhaustively.

\begin{exa}
Consider a tableau for the formula $\bo(p \lor q) \to (\bo p \lor \bo q)$ that begins with
\begin{center}
\begin{tabular}{rl}
{\footnotesize$1:$}	&	$[] > \lab{\bo(p \lor q) \to (\bo p \lor \bo q)}{1}$\\
{\footnotesize$2:$}	&	$r12$\\
{\footnotesize$3:$}	&	$\lab{\bo(p \lor q)}{1} > \lab{\bo p \lor \bo q}{1}$\\
{\footnotesize$4:$}	&	$\lab{\bo(p \lor q)}{1} > \lab{\bo p}{1}$\\
{\footnotesize$5:$}	&	$\lab{\bo(p \lor q)}{1} > \lab{\bo q}{1}$\\
{\footnotesize$6:$}	&	$\lab{\bo p}{1}  \geq \lab{p}{3}$\\
{\footnotesize$7:$}	&	$r13$\\
{\footnotesize$8:$}	&	$\lab{\bo q}{1}  \geq \lab{q}{4}$\\
{\footnotesize$9:$}	&	$r14$\\
{\footnotesize$10:$}	&	$\lab{p \lor q}{2} \geq \lab{\bo(p \lor q)}{1}$\\
{\footnotesize$11:$}	&	$\lab{p \lor q}{3} \geq \lab{\bo(p \lor q)}{1}$\\
{\footnotesize$12:$}	&	$\lab{p \lor q}{4} \geq \lab{\bo(p \lor q)}{1}$\\
\end{tabular}
\end{center}
then continues by splitting into two subtrees, namely\medskip
\begin{prooftree}\rootAtTop
\AxiomC{$\lab{p}{4} \geq \lab{\bo(p \lor q)}{1}$}
\AxiomC{$\lab{q}{4} \geq \lab{\bo(p \lor q)}{1}$}
\BinaryInfC{$\lab{p}{3} \geq \lab{\bo(p \lor q)}{1}$}
\AxiomC{$\lab{p}{4} \geq \lab{\bo(p \lor q)}{1}$}
\AxiomC{$\lab{q}{4} \geq \lab{\bo(p \lor q)}{1}$}
\BinaryInfC{$\lab{q}{3} \geq \lab{\bo(p \lor q)}{1}$}
\BinaryInfC{$\lab{p}{2} \geq \lab{\bo(p \lor q)}{1}$}
\end{prooftree}\medskip
 and a second that is exactly the same except that the root is $\lab{q}{2} \geq \lab{\bo(p \lor q)}{1}$.
 
 Observe now that the systems of inequations for the two leftmost branches of the subtree above are both inconsistent, since, combining inequations, we obtain
 \[
 \lab{p}{3} \geq \lab{\bo(p \lor q)}{1} > \lab{\bo p}{1} \geq \lab{p}{3}.
 \]
 Similarly, the system of inequations for the rightmost branch is inconsistent, since we obtain
 \[
 \lab{q}{4} \geq \lab{\bo(p \lor q)}{1} > \lab{\bo q}{1} \geq \lab{q}{4}.
 \]
 The  system of inequations for the remaining branch is consistent, however. Let us denote each $\lab{p}{i}$ and  $\lab{q}{i}$ by $x_i$ and $y_i$, respectively, for $i = 2,3,4$, $\lab{\bo p}{1}$ and $\lab{\bo q}{1}$ by $x'_1$ and $y'_1$, respectively, and $\lab{\bo(p \lor q)}{1}$ by $z$. Then for this branch, we obtain the set of inequations
 \[
 \{z > x'_1, \, z > y'_1, \, x'_1 \geq x_3, \, y'_1 \geq y_4, \, x_2 \geq z, \, y_3 \geq z, \, x_4 \geq z\}
 \]
 which can be satisfied over $\R$ by taking, e.g., 
 \[
 x_2 = 3, \,  x_3 = 0, \, x_4 = 3, \, y_3 = 3, \, y_4 = 0, \, x'_1 = 1,  \, y'_1 = 1,  \, z = 2.
 \]
 We obtain a $\lgc{K(A)}$-model $\fram{M} = \langle W, R, V \rangle$ by identifying $w_i$ in $W$ with each $i \in \N$ occurring on the branch and including $\langle w_i,w_j \rangle$ in $R$ whenever $rij$ appears; that is, $W = \{w_1,w_2,w_3,w_4\}$ and $R = \{\langle w_1,w_2 \rangle, \langle w_1,w_3 \rangle, \langle w_1,w_4 \rangle\}$. We also use the assignment satisfying the set of inequations to define (the other values are unimportant)
 \[
 V(p,w_2) = V(p,w_4) = V(q,w_3) = 3 \quad \text{and} \quad V(q,w_2) = V(p,w_3) = V(q,w_4) = 0.
 \]
 Then $V(\bo (p \lor q),w_1) = 3$ and $V(\bo p \lor \bo q,w_1) = 0$, so $V(\bo(p \lor q) \to (\bo p \lor \bo q),w_1) = -3$. 
\end{exa}

\subsection{Soundness}\label{ss:soundness}

Let $T$ be an $\lgc{LK(A)}$-tableau and let $B$ be a branch of $T$. We call a serial $\lgc{K(A)}$-model $\fram{M} = \langle W, R, V \rangle$  {\em faithful to $B$} if there is a map $f \colon \N \to W$ (said to {\em show} that $\fram{M}$ is faithful to $B$) such that if $rij$ occurs on $B$, then $Rf(i)f(j)$, and for every inequation $\lab{\f_1}{i_1},\ldots,\lab{\f_n}{i_n} \aineq \lab{\p_1}{j_1},\ldots,\lab{\p_m}{j_m}$ occurring on $B$,
\[
 V(\f_1,f(i_1)) + \ldots + V(\f_n,f(i_n)) \,\aineq\, V(\p_1,f(j_1)) + \ldots + V(\p_m,f(j_m)).
\]
We say that $\fram{M}$ is {\em faithful to $T$} if $\fram{M}$ is faithful to a branch $B$ of $T$. Observe that in this case, the map defined by $e(\lab{p}{i}) = V(p,i)$ and $e(\lab{\bo \f}{i}) = V(\bo \f, i)$ satisfies the system of inequations associated to $B$, and hence $T$ is open.  

The following lemma establishes the soundness of the rules of $\lgc{LK(A)}$.

\begin{lem}\label{LK:soundness lemma}
 Let $\fram{M} = \langle W, R, V \rangle$ be a finite serial $\lgc{K(A)}$-model faithful to a branch $B$ of an $\lgc{LK(A)}$-tableau $T$. If a rule of $\lgc{LK(A)}$ is applied to $B$, giving a tableau $T'$ extending $T$, then $\fram{M}$ is faithful to $T'$.
\end{lem}
\begin{proof}
 Let $f$ be a map showing that the finite serial $\lgc{K(A)}$-model $\fram{M}= \langle W, R, V \rangle$ is faithful to the branch $B$ of the tableau $T$. The cases of $\zerolr$, $\zerorr$, $\implr$, $\imprr$, $\multlr$, $\multrr$, $\lorr$, and $\wedl$ follow easily. For $\lorl$, suppose that $\lab{\Ga}{\vecn{k}}, \lab{\f \lor \p}{i} \aineq \lab{\De}{\vecn{l}}$ appears on $B$, and that we obtain an extension $T'$ of $T$ by two branches: one branch $B'$ extending $B$ with $\lab{\Ga}{\vecn{k}}, \lab{\f}{i} \aineq \lab{\De}{\vecn{l}}$, and another branch $B''$ extending $B$ with $\lab{\Ga}{\vecn{k}}, \lab{\p}{i} \aineq \lab{\De}{\vecn{l}}$. Let $\lab{\Ga}{\vecn{k}} = [\lab{\f_1}{k_1},\ldots,\lab{\f_n}{k_n}]$ and $\lab{\De}{\vecn{l}} = [\lab{\p_1}{l_1},\ldots,\lab{\p_m}{l_m}]$ and denote $V(\f_1,f(k_1)) + \ldots + V(\f_n,f(k_n))$ by $V(\Ga,f(\vecn{k}))$ and $V(\p_1,f(l_1)) + \ldots + V(\p_m,f(l_m))$ by $V(\De,f(\vecn{l}))$. Since $\fram{M}$ is faithful to $B$, we have  $V(\Ga,f(\vecn{k})) + V(\f \lor \p,f(i)) \aineq V(\De,f(\vecn{l}))$. Hence
  \[
V(\Ga,f(\vecn{k})) + \max(V(\f,f(i)), V(\p,f(i))) \,\aineq\, V(\De,f(\vecn{l})).
 \]
 If $\max(V(\f,f(i)), V(\p,f(i))) = V(\f,f(i))$ then $\fram{M}$ is faithful to the branch $B'$, otherwise $\fram{M}$ is faithful to the branch $B''$. Hence $\fram{M}$ is faithful to $T'$. The case of $\wedr$ follows similarly.
 
 For $\bolr$, suppose that $\lab{\Ga}{\vecn{k}}, \lab{\bo \f}{i} \aineq \lab{\De}{\vecn{l}}$ and $rij$ appear on $B$ and we obtain an extension $T'$ of $T$ by a branch $B'$ which extends $B$ with $\lab{\f}{j} \geq \lab{\bo\f}{i}$. Since $\fram{M}$ is faithful to $B$, we have $Rf(i)f(j)$. But then $V(\f,f(j)) \geq V(\bo \f,f(i))$, so $\fram{M}$ is faithful to $B'$ and~$T'$.

For $\borr$ suppose that $\lab{\Ga}{\vecn{k}} \aineq \lab{\bo \f}{i}, \lab{\De}{\vecn{l}}$ appears on $B$ and we obtain an extension $T'$ of $T$ by a branch $B'$ that extends $B$ with $rij$ ($j\in \N$ new) and $\lab{\bo\f}{i} \geq \lab{\f}{j}$. Since $\fram{M}$ is finite and serial, there exists $v\in W$ such that $Rf(i)v$ and $V(\bo \f,f(i)) = V(\f,v)$. Hence the map $f'$ defined to coincide with $f$ except that $f'(j)= v$ together with the branch $B'$ show that $\fram{M}$ is faithful to $T'$.  

Finally, for $\ex$, suppose that $rik$ appears on $B$ and we obtain an extension $T'$ of $T$ by a branch $B'$ that extends $B$ with $rkj$ ($j \in \N$ new). Since $rik$ is in $B$, we have $Rf(i)f(k)$. Because $\fram{M}$ is serial, there exists $v\in W$ such that $Rf(k)v$. The map $f'$ defined to coincide with $f$ except that $f'(j)= v$ shows that $\fram{M}$ is faithful to $B'$ and, hence, to $T'$.
\end{proof}

\begin{prop}\label{p:soundnesslabelledtableau}
If $\der{\lgc{LK(A)}} \f$, then $\mdl{\lgc{K(A)}} \f$.
\end{prop}
\begin{proof}
Suppose that $\not \mdl{\lgc{K(A)}} \f$. By Lemma~\ref{l:fmp}, there exist a finite serial $\lgc{K(A)}$-model $\fram{M} = \langle W, R, V \rangle$ and $w_1\in W$ such that $0 > V(\f, w_1)$. Let $f \colon \N \to W$ be any function such that $f(1) = w_1$ and $f(2) = w_2$, where $Rw_1w_2$. This function shows that $\fram{M}$ is faithful to the only branch of the tableau consisting just of the root $[] > [\lab{\f}{1}]$ and covering node $r12$. Suppose that by applying the decomposition rules to this tableau, we obtain a tableau $T$. Applying Lemma~\ref{LK:soundness lemma} inductively, $\fram{M}$ is faithful to $T$ by some branch $B$. But then  the system of inequations associated with $B$ is consistent over $\R$, and $T$ is open. Hence $\not \der{\lgc{LK(A)}} \f$.
\end{proof}

\vspace{-0.5em}
\subsection{Completeness} \label{ss:completeness}

We establish the completeness of $\lgc{LK(A)}$ by showing that an open branch of a tableau for a formula where the rules have been applied exhaustively generates a $\lgc{K(A)}$-model where the formula is not valid. In order to avoid repetitions occurring when a rule is applied more than once to a labelled inequation with the same conclusions (or with a new label in the case of $\borr$), we distinguish between active and inactive inequations and use new variables to denote modal formulas that have already been decomposed. To make this precise, we introduce the notation $\newv{\bo \f}$ to denote a variable corresponding to the modal $\lang_{\lgc{A}}^\bo$-formula $\bo \f$, and define $\newvar =  \var \cup \{\newv{\bo \f} : \bo \f \in \fm(\lang_{\lgc{A}}^\bo)\}$. We let $\fm(\lang_{\lgc{A}}^\bo)^*$ denote the set of $\lang_{\lgc{A}}^\bo$-formulas over $\newvar$, noting that of course $\fm(\lang_{\lgc{A}}^\bo) \subseteq \fm(\lang_{\lgc{A}}^\bo)^*$. The {\em complexity} of a labelled inequation $\lab{\Ga}{\vecn{k}} \aineq \lab{\De}{\vecn{l}}$ over $\fm(\lang_{\lgc{A}}^\bo)^*$ is defined as the sum of the complexities of the formula occurrences in $\Ga$ and $\De$.

We now consider a slight variant $\lgc{LK'(A)}$ of $\lgc{LK(A)}$, replacing the rules for $\bo$ with the following rules that decompose several occurrences of a labelled formula simultaneously:
\newcommand{\premisesborn}{
\begin{array}{c}
\lab{\Ga_1}{\vecn{k_1}} \aineq  n_1\lab{\bo \f}{i},  \lab{\De_1}{\vecn{l_1}}\\
\vdots\\
\lab{\Ga_m}{\vecn{k_m}} \aineq  n_m \lab{\bo \f}{i},\lab{\De_m}{\vecn{l_m}}\\
\end{array}
}
\newcommand{\conclusionsborn}{
\begin{array}{c}
\lab{\newv{\bo \f}}{i} \ge \lab{\f}{j} \\
rij \\
\lab{\Ga_1}{\vecn{k_1}} \aineq n_1 \lab{\newv{\bo \f}}{i}, \lab{\De_1}{\vecn{l_1}} \\
\vdots\\
\lab{\Ga_m}{\vecn{k_m}} \aineq n_m \lab{\newv{\bo \f}}{i}, \lab{\De_m}{\vecn{l_m}} \\
\end{array}}
\newcommand{\premisesboln}{
\begin{array}{c}
rij_1\\
\vdots\\
rij_s\\
\lab{\Ga_1}{\vecn{k_1}}, n_1\lab{\bo \f}{i} \aineq  \lab{\De_1}{\vecn{l_1}}\\
\vdots\\
\lab{\Ga_m}{\vecn{k_m}}, n_m \lab{\bo \f}{i} \aineq  \lab{\De_m}{\vecn{l_m}}\\
\end{array}}
\newcommand{\conclusionsboln}{
\begin{array}{c}
\lab{\f}{j_1} \ge \lab{\newv{\bo \f}}{i}\\
\vdots\\
\lab{\f}{j_s} \ge \lab{\newv{\bo \f}}{i}\\
\lab{\Ga_1}{\vecn{k_1}}, n_1 \lab{\newv{\bo \f}}{i} \aineq \lab{\De_1}{\vecn{l_1}} \\
\vdots\\
\lab{\Ga_m}{\vecn{k_m}}, n_m \lab{\newv{\bo \f}}{i} \aineq \lab{\De_m}{\vecn{l_m}} \\
\end{array}}
\[
\begin{array}{ccc}
\infer[\bolrn]{\conclusionsboln}{\premisesboln}
& &
\infer[\borrn]{\conclusionsborn}{\premisesborn}
\hspace{-1.25cm}j \in \N \text{ new}
\end{array}
\]
Closed and open $\lgc{LK'(A)}$-tableaux are defined as for $\lgc{LK(A)}$, except that the system associated to a branch of a tableau consists of all inequations on the branch that contain only variables from $\{\lab{q}{i} : q \in \newvar, i \in \N\}$. We call an $\lgc{LK'(A)}$-tableau for $\f \in \fm(\lang_{\lgc{A}}^\bo)$ {\em complete} if it is constructed as follows, making use of the notions of {\em active} and {\em inactive} inequations of the tableau to control applications of the rules:\smallskip

\begin{enumerate}

\item	Begin the tableau with the active labelled inequation $[] > [\lab{\f}{1}]$ and relation $r12$.\smallskip

\item	If all active labelled inequations have complexity $0$, then stop.\smallskip

\item	Apply the rules for $\0,\&,\to,\land,\lor$ exhaustively to active labelled inequations, changing the premise to inactive and the conclusions to active after each application.\smallskip

\item	Fix $i$ such that $\lab{\bo \p}{i}$ occurs in an active labelled inequation, and apply $\ex$ to every branch $B$ containing $\lab{\bo \chi}{i}$ for some $\bo \chi$ in an active inequation to obtain relations $rik_B$ for some new $k_B \in \N$.\smallskip

\item	For each $\lab{\bo \p}{i}$ occurring on the right in an active labelled inequation, apply $\borrn$ to the collection of all active labelled inequations $\lab{\Ga_t}{\vecn{k_t}} \aineq n_t \lab{\bo \p}{i}, \lab{\De_t}{\vecn{l_t}}$ (where $\lab{\bo \p}{i}$ does not occur in $\lab{\De_t}{\vecn{l_t}}$) on a branch, changing the premise to inactive and the conclusions to active after each application. \smallskip

\item	For each $\lab{\bo \p}{i}$ occurring on the left in an active labelled inequation, apply $\bolrn$ to the collection of all active labelled inequations $\lab{\Ga_t}{\vecn{k_t}}, n_t \lab{\bo \p}{i} \aineq  \lab{\De_t}{\vecn{l_t}}$ (where $\lab{\bo \p}{i}$ does not occur in $\lab{\Ga_t}{\vecn{k_t}}$) and all relations $rij_1,\ldots,rij_s$ on a  branch, changing the premises to inactive and the conclusions to active after each application. \smallskip

\item Repeat from (2).

\end{enumerate}

Observe that steps (3), (5), and (6) above decrease the multiset of complexities of the active labelled inequations, according to the standard multiset well-ordering (see~\cite{DeMa79}). Hence the procedure terminates with a complete $\lgc{LK'(A)}$-tableau $T$ for any $\f\in \fm(\lang_{\lgc{A}}^\bo)$.  Suppose now that we change each $\lab{\newv{\bo \p}}{i}$ to $\lab{\bo \p}{i}$ in $T$. Replacing applications of the rules $\bolrn$ and $\borrn$ with appropriate repeated applications of the rules $\bolr$ and $\borr$, we obtain an $\lgc{LK(A)}$-tableau $T'$ for $\f$  such that each branch of $T'$ contains all the inequations (modulo renaming of variables) occurring on the corresponding branch of $T$. Hence we obtain:

\begin{lem}\label{LK:equivalent calculus}
If there exists a closed complete $\lgc{LK'(A)}$-tableau for $\f \in \fm(\lang_{\lgc{A}}^\bo)$, then $\der{\lgc{LK(A)}}~\f$.
\end{lem}

 Let $T$ be an open complete $\lgc{LK'(A)}$-tableau for $\f \in \fm(\lang_{\lgc{A}}^\bo)$ and let $e$ be a map satisfying the system of inequations associated to an open branch $B$ of $T$.  We say that $\fram{M} = \langle W, R, V \rangle$ is the {\em $e$-induced model}  of $T$ by $B$ if 

\begin{itemize}

\item	$W = \{w_i : i \in \N \text{ is a label occurring on } B\}$;\smallskip

\item	$Rw_iw_j$ if and only if $rij$ occurs on $B$;\smallskip

\item	the valuation map $V \colon \newvar \times W \to [-r,r]$ is defined by
\[
\begin{array}{rcl}
V(p,w_i) & = & \begin{cases} e(\lab{p}{i}) & \text{ if } \lab{p}{i} \text{ occurs on } B\\
0 & \text{ otherwise} \end{cases}
\end{array}
\]
where $r =\max\{|e(\lab{p}{i})| : \lab{p}{i} \mbox{ occurs on } B\}$.

\end{itemize}

\begin{lem}\label{LK:completeness lemma}
Let $\fram{M} = \langle W, R, V \rangle$ be the $e$-induced model of an open complete $\lgc{LK'(A)}$-tableau $T$ by a branch $B$. Extend the map $e$ by fixing $e(\lab{\f}{i}) = V(\f,w_i)$ for each $\f \in \fm(\lang_{\lgc{A}}^\bo)^*$ and $w_i \in W$, and denote $e(\lab{\f_1}{k_1}) + \ldots + e(\lab{\f_n}{k_n})$ by $e((\Ga)^{\vecn{k}})$ for $\lab{\Ga}{\vecn{k}} = [\lab{\f_1}{k_1},\ldots,\lab{\f_n}{k_n}]$.  Then $e((\Ga)^{\vecn{k}}) \aineq e((\De)^{\vecn{l}})$ for each labelled inequation $\lab{\Ga}{\vecn{k}} \aineq \lab{\De}{\vecn{l}}$ that appears on $B$.
 \end{lem}

\begin{proof}
We prove the claim by induction on the complexity of $\lab{\Ga}{\vecn{k}} \aineq \lab{\De}{\vecn{l}}$. The base case follows using the definition of $\fram{M}$ and the fact that $e$ satisfies the system of inequations associated to $B$. Moreover, the cases where $\lab{\Ga}{\vecn{k}} \aineq \lab{\De}{\vecn{l}}$ appears as a premise of an application of a rule for $\0$, $\&$, or $ \to$ follow directly using the induction hypothesis. 

Suppose that the inequation is $\lab{\Ga'}{\vecn{k'}}, \lab{\f \lor \p}{i} \aineq \lab{\De}{\vecn{l}}$ and $\lab{\Ga'}{\vecn{k'}}, \lab{\f}{i} \aineq \lab{\De}{\vecn{l}}$ appears on $B$. (The case where $\lab{\Ga'}{\vecn{k'}}, \lab{\p}{i} \aineq \lab{\De}{\vecn{l}}$ appears on $B$ is symmetrical.) By the induction hypothesis, $e((\Ga')^{\vecn{k'}}) + e((\f)^i) \aineq e((\De)^{\vecn{l}})$. Since $e((\f \lor \p)^i) = \max(e((\f)^i),e((\p)^i))$ we obtain the desired inequality. The case when the inequation is $\lab{\Ga}{\vecn{k}} \aineq \lab{\f \land \p}{i}, \lab{\De'}{\vecn{l'}}$ follows similarly.

Suppose that the inequation is $\lab{\Ga}{\vecn{k}} \aineq \lab{\f \lor \p}{i}, \lab{\De'}{\vecn{l'}}$ and  $\lab{\Ga}{\vecn{k}} \aineq \lab{\f}{i}, \lab{\De'}{\vecn{l'}}$ and $\lab{\Ga}{\vecn{k}} \aineq \lab{\p}{i}, \lab{\De'}{\vecn{l'}}$ appear on $B$. The desired inequality follows  by applying the induction hypothesis to these two inequations and noticing that $e((\f \lor \p)^i) = \max(e((\f)^i),e((\p)^i))$. The case when the inequation is $\lab{\Ga'}{\vecn{k'}},  \lab{\f \land \p}{i} \aineq \lab{\De}{\vecn{l}}$ follows similarly.

Suppose that the inequation is $\lab{\Ga'}{\vecn{k'}}, n \lab{\bo \f}{i} \aineq \lab{\De}{\vecn{l}}$ and $\lab{\Ga'}{\vecn{k'}}, n\lab{\newv{\bo \f}}{i} \aineq \lab{\De}{\vecn{l}}$ occurs on $B$. Since $\fram{M}$ is finite and serial, there is a $j$ such that $rij$ occurs on $B$ and $V(\bo \f, w_i) = V(\f,w_j)$. But then also $\lab{\f}{j} \geq \lab{\newv{\bo \f}}{i}$ occurs on $B$. By the induction hypothesis twice,  $e((\Ga')^{\vecn{k'}}) + ne(\lab{\newv{\bo \f}}{i}) \aineq e((\De)^{\vecn{l}})$ and $e(\lab{\f}{j}) \geq e(\lab{\newv{\bo \f}}{i})$, and the desired inequality follows since also $e(\lab{\bo \f}{i})= V(\bo \f, w_i) = V(\f,w_j) = e(\lab{\f}{j})$. 

Finally, suppose that the inequation is $\lab{\Ga}{\vecn{k}} \aineq  n\lab{\bo \f}{i}, \lab{\De'}{\vecn{l'}}$ and $\lab{\Ga}{\vecn{k}} \aineq n\lab{\newv{\bo\f}}{i}, \lab{\De'}{\vecn{l'}}$ and $\lab{\newv{\bo\f}}{i} \geq \lab{\f}{j}$ appear on $B$ together with the relation $rij$. By the induction hypothesis twice, $e((\Ga)^{\vecn{k}})  \aineq ne(\lab{\newv{\bo \f}}{i}) + e((\De')^{\vecn{l'}})$ and $e(\lab{\newv{\bo \f}}{i}) \geq e(\lab{\f}{j})$, and the desired inequality follows since also $e(\lab{\f}{j}) = V(\f,w_j) \ge V(\bo\f,w_i) = e(\lab{\bo \f}{i})$.
\end{proof}

\begin{thm}\label{t:LabelledEquivSemantics}
The following are equivalent for any $\f \in \fm(\lang_{\lgc{A}}^\bo)$:
\begin{enumerate}[label=\rm (\arabic*)]
\item There exists a closed complete $\lgc{LK'(A)}$-tableau for $\f$.
\item	$\der{\lgc{LK(A)}} \f$.
\item $\mdl{\lgc{K(A)}} \f$.
\end{enumerate}
\end{thm}
\begin{proof}
(1)\,$\Rightarrow$\,(2)\,$\Rightarrow$\,(3) is just the combination of Lemma~\ref{LK:equivalent calculus} and Proposition~\ref{p:soundnesslabelledtableau}. We prove (3)\,$\Rightarrow$\,(1) by contraposition. If (1) fails, then there is an open complete $\lgc{LK'(A)}$-tableau $T$  beginning with  $[] > [\lab{\f}{1}], \ r12$. Let $e$ be a map satisfying the system of inequations associated to a branch $B$ of $T$ and consider the $e$-induced model $\fram{M} = \langle W, R, V \rangle$ of $T$ by $B$. By Lemma \ref{LK:completeness lemma}, we obtain $0 > e(\lab{\f}{1}) = V(\f,w_1)$. Hence $\not \mdl{\lgc{K(A)}} \f$.
\end{proof}

Let us remark here that there exist significant similarities between $\lgc{LK(A)}$ and the tableau calculus given in~\cite{KPS13} for the fuzzy description logic ``{\L}ukasiewicz fuzzy $\mathcal{ALC}$''. Both calculi reduce the validity of a formula to the satisfiability of linear programming problems, using labels to record values of formulas at different worlds. Superficial differences arise as a result of the restriction of values for {\L}ukasiewicz fuzzy $\mathcal{ALC}$ to the real unit interval $[0,1]$ and the use of several modal operators (corresponding to roles in the description logic). More significantly, roles in {\L}ukasiewicz fuzzy $\mathcal{ALC}$ are interpreted by fuzzy rather than crisp relations and appear also in inequations, whereas $\lgc{LK(A)}$ proceeds by directly generating a crisp frame suitable for constructing a potential countermodel.

\subsection{Complexity} \label{ss:complexity}

It follows directly from the completeness proof above that checking the $\lgc{K(A)}$-validity of an $\lang_\lgc{A}^\bo$-formula $\f$ is decidable. We simply apply the procedure for building a complete $\lgc{LK'(A)}$-tableau for $\f$ to generate finitely many linear programming problems which can then be checked for satisfiability. Considering this procedure in more detail, we obtain  an upper bound for the complexity of checking $\lgc{K(A)}$-validity.

\begin{thm}\label{t:complexity}
The problem of checking if $\f \in \fm(\lang_\lgc{A}^\bo)$ is $\lgc{K(A)}$-valid is in {\sc coNEXPTIME}.
\end{thm}
\begin{proof}
By Theorem~\ref{t:LabelledEquivSemantics}, we may consider a complete $\lgc{LK'(A)}$-tableau $T$ for an $\lang_\lgc{A}^\bo$-formula $\f$  obtained by following steps (1)-(7) in the procedure above. We may also assume that no labelled inequation appears twice on the same branch of $T$. Suppose that $\f$ has complexity $n$. A new label $j$ is introduced by applying the rule $\borrn$ to a labelled inequation $\lab{\Ga}{\vecn{k}} \aineq \lab{\bo \p}{i}, \lab{\De}{\vecn{l}}$, and by step (4), producing a new labelled inequation $\lab{\newv{\bo \p}}{i} \ge \lab{\p}{j}$, where $\bo \p$ is a subformula of $\f$, and $\p$ has smaller modal depth than $\bo \p$. Note that the number of subformulas $\bo \p$ of $\f$ is bounded by $n$; also the modal depth of $\f$ is bounded by $n$.  Hence the number of labels appearing on a branch of $T$ is at most exponential in $n$. Observe next that the complexity of any labelled inequation that occurs in $T$ is bounded by $n$, and that there are at most $n$ new variables of the form $\newv{\bo \p}$ appearing in $T$. Hence the number of different labelled inequations that can appear in $T$, and so also the length of any branch of $T$, is at most exponential in $n$.

To show that $\f$ is not $\lgc{K(A)}$-valid, we choose a branch $B$ of $T$ non-deterministically, noting that (binary) branching occurs only when applying the rules $\lor$ and $\land$. By the above reasoning, the length of $B$ and the complexity of the labelled inequations appearing on $B$ are at most exponential in $n$. The result then follows from the fact that the linear programming problem is in {\sc P}~\cite{Khachiyan79}. 
\end{proof}

It is no surprise that the upper bound provided here for checking $\lgc{K(A)}$-validity matches the known upper bound for checking validity in fuzzy description logics based on infinite-valued \L ukasiewicz logic (see~\cite{KPS13}) and indeed also the \L ukasiewicz modal logic described in Section~\ref{s:modalabelianlogic}. In all these cases, unpacking the semantics leads to a non-deterministic guessing of linear programming problems of exponential size in the complexity of the original formula. Validity in modal or description logics based on finite \L ukasiewicz logics is known to be PSPACE-complete~\cite{BCE11}, and the same holds for many-valued modal logics based on G{\"o}del logics~\cite{CMRR17}; however, these arguments do not seem to generalize to the current setting.


\section{The Modal-Multiplicative Fragment} \label{s:fragment}

In this section, we provide an axiom system (without infinitary rules) and analytic sequent calculus for the modal-multiplicative fragment of $\lgc{K(A)}$, and in doing so, take a first step towards obtaining such systems for the full logic. 


\subsection{An Axiom System}\label{ss:axiomatization}

\begin{figure}[tbp] 
\centering
\fbox{
\begin{minipage}{10 cm}
\[
\begin{array}{rll}
{\rm (B)}  & \	& (\f  \to \p) \to ((\p \to \x) \to (\f \to \x))\\[.025in]
{\rm (C)}  & 	& (\f \to (\p \to \x)) \to (\p \to (\f \to \x))\\[.025in]
{\rm (I)}    &		& \f \to \f\\[.025in]
{\rm (A)}  & 	& ((\f \to \p) \to \p) \to \f\\[.025in]
{\rm (K)}  &		& \bo (\f \to \p) \to (\bo \f \to \bo \p)\\[.025in]
\textrm{(D$_n$)} &	& \bo(n\f) \to n\bo \f \qquad (n \ge 2)\\[.025in]
\end{array}
\]
\[
\infer[\textrm{(mp)}]{\p}{\f & \f \to \p}
\qquad
\infer[\textrm{(nec)}]{\bo \f}{\f}
\qquad
\infer[\textrm{(con$_n$) }]{\f}{n \f}
\quad
(n \ge 2)
\]
\caption{The axiom system $\lgc{K(A_m)}$}\label{f:kz}
\end{minipage}}
\end{figure}

For convenience (in particular, to reduce the number of cases in proofs), we define the modal-multiplicative fragment here over a language $\lang_{\lgc{A_m}}^\bo$ consisting of the binary connective $\to$ and unary connective $\bo$. To define further connectives, we fix $p_0 \in \var$ and let
\[
\0 := p_0 \to p_0, \quad 
\lnot \f := \f \to \0, \quad
\f \& \p := \lnot \f \to \p, \quad 
\text{and} \quad
\di \f := \lnot \bo \lnot \f.
\]
We also define $0\f := \0$ and $(n+1)\f := \f \& (n\f)$ for each $n \in \N$.

Our axiom system $\lgc{K(A_m)}$ for the modal-multiplicative fragment of $\lgc{K(A)}$ is presented in Figure~\ref{f:kz}. For a formula  $\f \in \fm(\lang_\lgc{A_m}^\bo)$, we write $\der{K(A_m)} \f$ if there exists a $\lgc{K(A_m)}$-derivation of $\f$, defined as usual as a finite sequence of $\lang_\lgc{A_m}^\bo$-formulas that ends with $\f$ and is constructed inductively using the axioms and rules of $\lgc{K(A_m)}$. 

Establishing soundness for this system is straightforward. It is easily checked that the axioms (B), (C), (I), (A), and (K) are valid in all $\lgc{K(A)}$-models. For the less standard axioms (D$_n$) ($n \ge 2$), it suffices to consider a $\lgc{K(A)}$-model $\fram{M} = \langle W, R,V \rangle$ and $x \in W$, and observe that for all  $\f \in \fm(\lang_\lgc{A_m}^\bo)$,
\[
\begin{array}{rcl}
V(\bo(n\f),x) & = & \bigwedge_{\R} \{V(n\f, y) : Rxy \}\\[.05in]
				& = & \bigwedge_{\R} \{nV(\f, y) : Rxy \}\\[.05in]
				& = &  n\bigwedge_{\R} \{V(\f, y) : Rxy \}\\[.05in]
				& = & V(n\bo \f,x).
\end{array}
\]
It is clear that (mp) and (nec) preserve validity in $\lgc{K(A)}$-models. For (con$_n$) ($n \ge 2$), we just note that if $V(n\f,x) \ge 0$ for a  $\lgc{K(A)}$-model $\fram{M} = \langle W, R,V \rangle$ and $x \in W$, then $V(\f,x) \ge 0$. Hence a simple induction on the length of a $\lgc{K(A_m)}$-derivation gives the following result.
 
\begin{prop}\label{p:axiomsystemsound}
Let $\f \in \fm(\lang_{\lgc{A_m}}^\bo)$. If $\der{K(A_m)} \f$, then $\mdl{\lgc{K(A)}} \f$.
\end{prop}

\noindent
The proof of the converse direction is much harder. Before arriving finally at this result in Theorem~\ref{t:main}, we make a detour via a sequent calculus and exploit the completeness result for our labelled tableau calculus provided by Theorem~\ref{t:LabelledEquivSemantics}.


\subsection{A Sequent Calculus}\label{ss:sequent}

For the purposes of this paper, a {\em sequent} is an ordered pair of finite multisets of $\lang_{\lgc{A_m}}^\bo$-formulas $\Ga$ and $\De$, written $\Ga \seq \De$. A {\em sequent rule} is a set of {\em instances}, each consisting of a finite set of sequents called {\em premises} and a sequent called the {\em conclusion}. Such rules are typically written schematically, using $\f,\p,\x$ and $\Ga,\De,\Pi,\Si$  to denote arbitrary formulas and finite multisets of formulas, respectively. We also often write $\Ga,\De$ to denote the multiset union $\Ga \uplus \De$, $n\Ga$ for $\Ga,\ldots,\Ga$ ($n$ times), and $\bo \Ga$ for $[\bo \f : \f \in \Ga]$. 

We make use of a formula translation (assuming  $\f_1 \& \ldots \& \f_n = \0$ for $n=0$),
\[
\begin{array}{rcl}
\mathcal{I}(\f_1,\ldots,\f_n \seq \p_1,\ldots,\p_m) & := &  (\f_1 \& \ldots \& \f_n) \to (\p_1 \& \ldots \& \p_m),
\end{array}
\]
and say that a sequent $\Ga \seq \De$ is \emph{$\lgc{K(A)}$-valid}, written $\mdl{\lgc{K(A)}} \Ga \seq \De$, if $\mdl{\lgc{K(A)}}\mathcal{I}(\Ga \seq \De)$.

A {\em sequent calculus} $\lgc{GL}$ consists of a set of sequent rules, and a {\em $\lgc{GL}$-derivation} of a sequent $S$ from a set of sequents $Y$ is a finite tree of sequents with root $S$ such that each node is either (i) a leaf node and in $Y$,  or (ii) together with its parent nodes forms an instance of a rule of $\lgc{GL}$. In this case, we write $Y \der{\lgc{GL}} S$ or just $\der{\lgc{GL}} S$ if $Y = \emptyset$. A sequent rule is $\lgc{GL}$-{\em derivable} if there is a $\lgc{GL}$-derivation of the conclusion of any instance of the rule from its premises; $\lgc{GL}$-{\em admissible} if whenever the premises of an instance of the rule are $\lgc{GL}$-derivable, the conclusion is  $\lgc{GL}$-derivable; and  $\lgc{GL}$-{\em invertible} if whenever the conclusion of an instance of the rule is $\lgc{GL}$-derivable, the premises are $\lgc{GL}$-derivable. 

\begin{figure}[tbp] 
\centering
\fbox{
\begin{minipage}{14.75 cm}
 \[
 \begin{array}{c}
 \begin{array}{ccc}
 \infer[\idr]{\De \seq \De}{} & & 
 \infer[\cutr]{\Ga, \Pi \seq \Si, \De}{\Ga, \f \seq \De & \Pi \seq \f, \Si}\\[.15in]
 \infer[\mixr]{\Ga,\Pi \seq \Si,\De}{\Ga \seq \De & \Pi \seq \Si}  & \quad & 
 \qquad  \qquad  \infer[\seqcontn]{\Ga \seq \De}{n\Ga \seq n\De} \quad (n \ge 2)\\[.15in]
 \infer[\ilr]{\Ga, \f \to \p \seq \De}{\Ga, \p \seq \f, \De} &  & \infer[\irr]{\Ga \seq \f \to \p, \De}{\Ga, \f \seq \p, \De}\\[.15in]
 \end{array}\\
\infer[\boxknr{n}]{\bo \Ga \seq n[\bo \f]}{\Ga \seq n[\f]}\quad (n \ge 0) \quad 
\end{array}
 \]
 \caption{The sequent calculus $\lgc{GK(A_m)}$} \label{f:gkz}
\end{minipage}}
\end{figure}

A sequent calculus $\lgc{GK(A_m)}$ for the modal-multiplicative fragment of $\lgc{K(A)}$, an extension of a calculus for the multiplicative fragment of Abelian logic given in~\cite{met:seq}, is presented in Figure~\ref{f:gkz}. Although only rules for $\to$ and $\bo$ appear in this system, the following rules for other connectives  are $\lgc{GK(A_m)}$-derivable:\smallskip
\[
\begin{array}{ccc}
\infer[\salr]{\Ga, \f\&\p \seq \De}{\Ga, \f, \p \seq \De} & \qquad \qquad & 
\infer[\sarr]{\Ga \seq \f\&\p, \De}{\Ga \seq \f, \p, \De}\\[.1in]
\infer[\nlr]{\Ga, \lnot \f \seq \De}{\Ga \seq \f, \De} & & 
\infer[\nrr]{\Ga \seq\lnot \f, \De}{\Ga, \f \seq \De}\\[.1in]
\infer[\zrlr]{\Ga, \0 \seq \De}{\Ga \seq \De} & & 
\infer[\zrrr]{\Ga \seq \0, \De}{\Ga \seq \De}
\end{array}
\]

\begin{exa}\label{ex:dn}
Below we provide a simple example of a $\lgc{GK(A_m)}$-derivation, making use of the derived rules for $\&$ given above.
\[
\infer[\pfa{\irr}]{\seq \bo (\f \& \f) \to (\bo \f \& \bo \f)}{
 \infer[\pfa{\sarr}]{\bo (\f \& \f) \seq \bo \f \& \bo \f}{
  \infer[\pfa{\boxknr{2}}]{\bo (\f \& \f) \seq \bo \f, \bo \f}{
   \infer[\pfa{\salr}]{\f \& \f \seq \f, \f}{
     \infer[\pfa{\idr}]{\f,\f \seq \f,\f}{}}}}}
\]
Sequents of the form $\seq \bo (n\f) \to n \bo \f$ can be proved similarly using the rule $\boxknr{n}$.
\end{exa}

It is straightforward to establish an equivalence between derivability of a sequent in $\lgc{GK(A_m)}$ and derivability of its formula interpretation in the axiom system $\lgc{K(A_m)}$.

\begin{prop}\label{p:sequentcalculusaxiomsystemequivalent}
$\der{\lgc{GK(A_m)}} \Ga \seq \De$ if and only if $\der{\lgc{K(A_m)}} \mathcal{I}(\Ga \seq \De)$.
\end{prop}
\begin{proof}
It suffices for the left-to-right direction to show that for any rule of $\lgc{GK(A_m)}$ with premises $S_1,\ldots,S_m$ and conclusion $S$,  whenever $\der{\lgc{K(A_m)}} \mathcal{I}(S_i)$ for each $i \in \{1, \ldots, m\}$, also $\der{\lgc{K(A_m)}} \mathcal{I}(S)$. For example, consider the rule $\boxknr{n}$ and assume that $\der{\lgc{K(A)}} \mathcal{I}(\Ga \seq n[\f])$. Suppose that $\Ga = [\p_1, \ldots, \p_m]$ and let $\p = \p_1 \& \ldots \& \p_m$.  We continue the $\lgc{K(A_m)}$-derivation of $\mathcal{I}(\Ga \seq n[\f]) =\p \to n\f$ to obtain a $\lgc{K(A_m)}$-derivation of $\bo \p \to n\bo \f$:
\[
\begin{array}{rll}
  1.\, & \p \to n\f\\
  2.\, & \bo (\p \to n\f) &(\text{nec})\\
  3.\, & \bo (\p \to n\f) \to (\bo \p \to \bo n\f) &  \rm{(K)}\\
  4.\, & \bo \p \to \bo n\f & (\text{mp}) \text{ with 2,3}\\
  5.\, & \bo n\f \to n\bo \f & \textrm{(D$_n$)}\\
  6.\, &  (\bo \p \to \bo n\f) \to ((\bo n\f \to n\bo \f) \to (\bo \p \to n\bo \f)) \quad & \rm{(B)}\\
  7.\, & (\bo n\f \to n\bo \f) \to (\bo \p \to n\bo \f) &  (\text{mp}) \text{ with 4,6}\\
  8.\, & \bo \p \to n\bo \f & (\text{mp}) \text{ with 5,7}. 
  \end{array}
 \]
$(\bo \p_1 \& \ldots \& \bo \p_m) \to \bo \p$ is derivable using (B), (C), (I), and (K),  so, using (B) and (mp), we obtain  a $\lgc{K(A_m)}$-derivation of  $\mathcal{I}(\bo \Ga \seq n[\bo \f]) = (\bo \p_1 \& \ldots \& \bo \p_m) \to n \bo \f$.
 
For the right-to-left direction, it is easy to show that every axiom of $\lgc{K(A_m)}$ is $\lgc{GK(A_m)}$-derivable; see, e.g., Example~\ref{ex:dn} for $\lgc{GK(A_m)}$-derivations of instances of (D$_n$). Also, the rules of $\lgc{K(A_m)}$ are $\lgc{GK(A_m)}$-derivable. For example, for $(\text{con}_n)$, starting with $\seq n\f$, we apply  $\cutr$ with the $\lgc{GK(A_m)}$-derivable sequent $n\f \seq n[\f]$ to obtain $\seq n[\f]$ and then, applying $\seqcontn$, also $\seq \f$. Hence, if $\der{\lgc{K(A_m)}} \mathcal{I}(\Ga \seq \De)$, then  $\der{\lgc{GK(A_m)}} \seq \mathcal{I}(\Ga \seq \De)$ and, applying $\cutr$ with the $\lgc{GK(A_m)}$-derivable sequent $\Ga, \mathcal{I}(\Ga \seq \De) \seq \De$, also $\der{\lgc{GK(A_m)}} \Ga \seq \De$. 
 \end{proof}
 
We now consider a more complicated family of rules, indexed by $k \in \N \setminus\!\{0\}$ and $n \in \N$, that will be very useful in subsequent cut-elimination and completeness proofs:\smallskip
\[
\begin{array}{c}
\infer[\boxknr{k,n} \quad ]{\De, \bo \Ga \seq \bo \f_1,\ldots, \bo \f_n, \De}{
\Ga_0 \seq & \Ga_1 \seq k[\f_1] & \ldots & \Ga_n \seq k[\f_n]}\quad
 \text{where } k\Ga = \Ga_0 \uplus \Ga_1 \uplus \ldots \uplus \Ga_n.
\end{array}
\]
Critically for our later considerations, $\boxknr{k,n}$ is $\lgc{GK(A_m)}$-derivable for all $k \in \N \setminus\!\{0\}$, $n \in \N$ (for $k = 1$, omitting the application of $\seqcontk$):
\[
\infer[\pfa{\mixr}]{\De, \bo \Ga \seq \bo \f_1,\ldots, \bo \f_n, \De}{
  \infer[\pfa{\idr}]{\De \seq \De}{} &
 \infer[\pfa{\seqcontk}]{\bo\Ga \seq \bo \f_1,\ldots, \bo \f_n}{
  \infer[\pfa{\mixr}]{\bo (\Ga_0 \uplus \Ga_1 \uplus \ldots \uplus \Ga_n) \seq k [\bo \f_1],\ldots, k[\bo \f_n]}{
    \infer[\pfa{\boxknr{0}}]{\bo \Ga_0 \seq}{ 
     \Ga_0 \seq} & 
   \infer[\pfa{\mixr}]{\bo (\Ga_1 \uplus \ldots \uplus \Ga_n) \seq k [\bo \f_1],\ldots, k[\bo \f_n]}{
    \infer[\pfa{\boxknr{k}}]{\bo \Ga_1 \seq k[\bo \f_1]}{
     \Ga_1 \seq k [\f_1]} &
       \infer[\pfa{\mixr}]{\vdots}{
        \infer[\pfa{\boxknr{k}}]{\bo \Ga_n \seq k[\bo \f_n]}{
     \Ga_n \seq k [\f_n]}}}}}}
       \]

We devote the remainder of this subsection to showing that the calculus $\lgc{GK(A_m)}$ admits cut-elimination. That is, we provide an algorithm for constructively eliminating applications of the rule $\cutr$ from $\lgc{GK(A_m)}$-derivations. Observe first that the ``cancellation'' rule 
\[
 \infer[\canr]{\Ga \seq \De}{\Ga, \f \seq \f, \De}
\]
is both $\lgc{GK(A_m)}$-derivable  and can be used, with $\mixr$, to derive $\cutr$:\smallskip
\[
\infer[\pfa{\cutr}]{\Ga \seq \De}{
 \infer[\pfa{\ilr}]{\f \to \f \seq}{
  \infer[\pfa{\idr}]{\f \seq \f}{}}
&
  \infer[\pfa{\irr}]{\Ga \seq \f \to \f, \De}{
   \Ga, \f \seq \f, \De}}
\qquad
\infer[\pfa{\canr}]{\Ga, \Pi \seq \Si, \De}{
 \infer[\pfa{\mixr}]{\Ga, \Pi, \f \seq \f, \Si, \De}{
  \Ga, \f \seq \De & \Pi \seq \f, \Si}}
\]
Hence, to prove cut-elimination, it will be enough to show constructively that $\canr$ is admissible in $\lgc{GK(A_m)}$ without $\cutr$.

We begin by showing that every cut-free $\lgc{GK(A_m)}$-derivation can be transformed into a derivation in a restricted calculus $\lgc{GK(A_m)^r}$ consisting only of the rules $\idr$, $\ilr$, $\irr$, and $\boxknr{k,n}$ ($k \in \N \setminus\!\{0\}$, $n \in \N$). 

\begin{lem}\label{l:invertiblerules}
The rules $\ilr$ and $\irr$ are $\lgc{GK(A_m)^r}$-invertible.
\end{lem}
\begin{proof}
To show that $\ilr$ is $\lgc{GK(A_m)^r}$-invertible, we prove, more generally, that $\der{\lgc{GK(A_m)^r}} \Ga, m[\f \to \p] \seq \De$ implies $\der{\lgc{GK(A_m)^r}} \Ga, m\p \seq m\f, \De$ for all $m \in \N$, proceeding by induction on the height of a $\lgc{GK(A_m)^r}$-derivation of $\Ga, m[\f \to \p] \seq \De$. For the base case, $\De = \Ga \uplus m[\f \to \p]$ and it suffices to observe that $\der{\lgc{GK(A_m)^r}} \Ga, m\p \seq m\f, m[\f \to \p], \Ga$. For the inductive step, we observe that when the last rule applied is  $\ilr$ or $\irr$, the claim follows immediately by applying the induction hypothesis and, where necessary, the relevant rule. If the last rule applied is $\boxknr{k,n}$, then $m[\f \to \p]$ must occur also on the right of the sequent and the claim follows by first applying the rule $\boxknr{k,n}$ and then $\irr$ $m$ times. The proof that $\irr$ is $\lgc{GK(A_m)^r}$-invertible is very similar.
\end{proof}

\begin{lem}\label{l:admissiblemixsck}
The rules $\mixr$ and $\seqcontn$ are $\lgc{GK(A_m)^r}$-admissible.
\end{lem}
\begin{proof}
To show the $\lgc{GK(A_m)^r}$-admissibility of $\mixr$, we prove that
\[
\der{\lgc{GK(A_m)^r}} \Ga \seq \De \ \text{ and } \ \der{\lgc{GK(A_m)^r}} \Pi \seq \Si \quad \Longrightarrow \quad
\der{\lgc{GK(A_m)^r}} r\Ga, s\Pi \seq s\Si,r\De \ \text{ for all }  r,s \in \N,
\]
proceeding by induction on the sum of the heights of $\lgc{GK(A_m)^r}$-derivations $d_1$ and $d_2$ of $\Ga \seq \De$ and $\Pi \seq \Si$, respectively.
  
For the base case, if $d_1$ and $d_2$ have height $0$, then $\Ga \seq \De$ and $\Pi \seq \Si$ are instances of $\idr$, i.e., $\Ga = \De$ and $\Pi = \Si$. So $r\Ga \uplus s\Pi = r\De\uplus s\Si$ and $\der{\lgc{GK(A_m)^r}} r\Ga, s\Pi \seq s\Si,r\De$ by $\idr$.  If the last application of a rule in $d_1$ or $d_2$ is $\ilr$ or $\irr$, then the result follows easily by an application of the induction hypothesis and further applications of the rule.

Suppose now that $d_1$ ends with 
\[
\begin{array}{c}
\infer[\boxknr{k,n}]{\Omega, \bo \Ga' \seq \bo \f_1,\ldots, \bo \f_n, \Omega}{
\Ga_0 \seq & \Ga_1 \seq k[\f_1] & \ldots & \Ga_n \seq k[\f_n]}
\quad
 \text{where }  k\Ga' = \Ga_0 \uplus \Ga_1 \uplus \ldots \uplus \Ga_n.
\end{array}
\]
If $d_2$ has height $0$, then $\Pi = \Si$. An application of the induction hypothesis to the $\lgc{GK(A_m)^r}$-derivation of the premise $\Ga_0 \seq$ together with a $\lgc{GK(A_m)^r}$-derivation of the empty sequent $\seq$ of height $0$ yields $\der{\lgc{GK(A_m)^r}} r\Ga_0 \seq$. It follows then that the sequent $r\Omega, r\bo \Ga', s\Pi \seq s\Pi, r\bo \f_1,\ldots, r\bo \f_n, r\Omega$ is $\lgc{GK(A_m)^r}$-derivable using an application of the rule $\boxknr{k,rn}$. The case where $d_1$ has height $0$ and $d_2$ ends with $\boxknr{k,n}$ is symmetrical.

If $d_2$ ends with 
 \[
\begin{array}{c}
 \infer[\boxknr{l,m}]{\Theta, \bo \Pi' \seq \bo \p_1,\ldots, \bo \p_m, \Theta}{
\Pi_0 \seq & \Pi_1 \seq l[\p_1] & \ldots & \Pi_m \seq l[\p_m]} 
\quad \text{where }  l\Pi' = \Pi_0 \uplus \Pi_1 \uplus \ldots \uplus \Pi_m,
\end{array}
\]
then we obtain the required $\lgc{GK(A_m)^r}$-derivation 
 \[
 \infer[\boxknr{kl,rn+sm}]{r\Omega, s\Theta, r\bo \Ga', s\bo \Pi' \seq r\bo \f_1,\ldots, r\bo \f_n, s\bo \p_1,\ldots, s\bo \p_m, r\Omega, s\Theta}{
rl\Ga_0, sk\Pi_0 \seq &  
\{l\Ga_i \seq kl[\f_i]\}_{i \in \{1,\ldots,n\}} & \{k\Pi_j \seq kl[\p_j]\}_{1 \le j \le m}}
 \]
 where the premises are all $\lgc{GK(A_m)^r}$-derivable using the induction hypothesis.
 
We establish the $\lgc{GK(A_m)^r}$-admissibility of $\seqcontn$ by proving that
\[
\der{\lgc{GK(A_m)^r}} n\Ga \seq n\De \quad \Longrightarrow \quad \der{\lgc{GK(A_m)^r}} \Ga \seq \De,
\]
proceeding by induction on the sum of the complexities of the formulas in $\Ga,\De$. For the base case, if $n\Ga = n\De$ (in particular if $\Ga$ and $\De$ contain only variables), then $\Ga=\De$ and $\der{\lgc{GK(A_m)^r}} \Ga \seq \De$ by $\idr$. If $\Ga$ contains a formula $\f \to \p$, then by the invertibility of the rule $\ilr$ established in Lemma~\ref{l:invertiblerules}, $\der{\lgc{GK(A_m)^r}} n(\Ga - [\f \to \p]), n\p \seq n\f, n\De$. The induction hypothesis and an application of $\ilr$ gives $\der{\lgc{GK(A_m)^r}} \Ga \seq \De$. The  case where $\De$ contains a formula $\f \to \p$ is symmetrical. In the final case, the $\lgc{GK(A_m)^r}$-derivation of $n\Ga \seq n\De$ must end with an application of $\boxknr{k,nl}$ where $\Ga = \Pi \uplus [\bo \Si]$ and $\De = \Pi \uplus [\bo \f_1,\ldots,\bo \f_l]$. Hence $\der{\lgc{GK(A_m)^r}} \Ga \seq \De$ using $\boxknr{kn,l}$ and  the $\lgc{GK(A_m)^r}$-admissibility of $\mixr$.
\end{proof}

We now have all the necessary tools to prove the promised cut-elimination theorem.

\begin{thm}\label{t:cutelimination}
$\lgc{GK(A_m)}$ admits cut-elimination. 
\end{thm}
\begin{proof}
To establish cut-elimination for $\lgc{GK(A_m)}$, it suffices to prove that an uppermost application of $\cutr$ in a $\lgc{GK(A_m)}$-derivation can be eliminated; that is, we show that cut-free $\lgc{GK(A_m)}$-derivations of the premises of an instance of $\cutr$ can be transformed into a cut-free $\lgc{GK(A_m)}$-derivation of the conclusion. Observe first that the rule $\boxknr{n}$ is $\lgc{GK(A_m)^r}$-derivable using $\boxknr{k,n}$ with $k = n$, $\f_1 = \ldots = \f_n = \f$, and $\Ga_1 = \ldots = \Ga_n = \Ga$. Hence, the proof of Lemma~\ref{l:admissiblemixsck} shows that any cut-free $\lgc{GK(A_m)}$-derivation can be transformed algorithmically into a $\lgc{GK(A_m)^r}$-derivation. We prove (constructively) that
\[
\der{\lgc{GK(A_m)^r}} \Ga, \f \seq \f, \De \quad \Longrightarrow \quad \der{\lgc{GK(A_m)^r}} \Ga \seq \De. \qquad (\star)
\]
Suppose then that there are cut-free $\lgc{GK(A_m)}$-derivations of the premises  $\Ga, \f \seq \De$ and $\Pi \seq \f, \Si$ of an uppermost application of $\cutr$. By $\mixr$, we obtain a cut-free $\lgc{GK(A_m)}$-derivation of $\Ga, \Pi, \f \seq \f, \Si, \De$  and hence a $\lgc{GK(A_m)^r}$-derivation of this sequent. By $(\star)$, we obtain a $\lgc{GK(A_m)^r}$-derivation of  $\Ga, \Pi \seq \Si, \De$, which also gives the desired cut-free $\lgc{GK(A_m)}$-derivation.

We prove $(\star)$ by induction on the lexicographically ordered pair consisting of the modal depth of $\f$ and the sum of the complexities of the formulas in $\Ga, \f \seq \f, \De$. If $\Ga \uplus [\f] = [\f] \uplus \De$, then $\Ga = \De$ and $\Ga \seq \De$ is derivable using $\idr$. If $\f$ has the form $\p \to \x$, then we use the $\lgc{GK(A_m)^r}$-invertibility of $\ilr$ and $\irr$ and apply the induction hypothesis twice. The cases where $\Ga$ or $\De$ includes a formula $\p \to \x$ are very similar. Lastly, suppose that $\Ga, \f \seq \f, \De$ contains only variables and box formulas. Then there is a $\lgc{GK(A_m)}$-derivation of the sequent ending with an application of $\boxknr{k,n}$. The case where $\f$ is a variable is trivial, so let us just consider the case where $\f = \bo \x$ and the derivation ends with an application of $\boxknr{k,n}$. The case where $\f$ occurs in the context appearing on both sides of the conclusion follows immediately, so suppose that the derivation ends with
\[
\infer[\boxknr{k,n}]{\Si, \bo \Pi, \bo \x \seq \bo \x, \bo \p_2,\ldots, \bo \p_n, \Si}{
\Pi_0, k_0[\x] \seq & \Pi_1, k_1[\x] \seq k[\x] & \{\Pi_i, k_i[\x] \seq k[\p_i]\}_{i=2}^{n}}
\]
where $k\Pi = \Pi_0 \uplus  \Pi_1 \uplus \ldots \uplus \Pi_n$ and $k = k_0 + k_1 + \ldots + k_n$.  By the induction hypothesis, 
\[
\der{\lgc{GK(A_m)^r}} \Pi_1 \seq (k-k_1)[\x].
\]
By Lemma~\ref{l:admissiblemixsck} (the $\lgc{GK(A_m)^r}$-admissibility of $\mixr$), we have $\lgc{GK(A_m)^r}$-derivations of
\[
\begin{array}{l}
 k_0 \Pi_1,  (k-k_1)\Pi_0, (k-k_1)k_0[\x] \seq (k-k_1)k_0[\x]\\[.025in]
 k_i \Pi_1,  (k-k_1)\Pi_i, (k-k_1)k_i[\x] \seq (k-k_1)k_i[\x], (k-k_1)k[\p_i]\quad \text{for } i \in \{2,\ldots, n\}.
\end{array}
\]
So, by the induction hypothesis, we have $\lgc{GK(A_m)^r}$-derivations of
\[
\begin{array}{l}
k_0 \Pi_1,  (k-k_1)\Pi_0 \seq \\[.025in]
k_i \Pi_1,  (k-k_1)\Pi_i \seq (k-k_1)k[\p_i] \quad \text{for }  i \in \{2,\ldots, n\}.
\end{array}
\]
Now by an application of  $\boxknr{(k-k_1)k,n-1}$, we have a $\lgc{GK(A_m)^r}$-derivation ending with
\[
\infer{\Si, \bo \Pi \seq \bo \p_2,\ldots, \bo \p_n, \Si}{
k_0 \Pi_1,  (k-k_1)\Pi_0 \seq  &  \{k_i \Pi_1,  (k-k_1)\Pi_i \seq (k-k_1)k[\p_i]\}_{i=2}^n}
\]
where $(k-k_1)k\Pi = (k_0 + k_2 + \ldots + k_n)(\Pi_0 \uplus  \Pi_1 \uplus \ldots \uplus \Pi_n)$. 
\end{proof}


\subsection{Completeness}

In this section we establish the completeness of both the axiom system $\lgc{K(A_m)}$ and the sequent calculus $\lgc{GK(A_m)}$ for the modal-multiplicative fragment of $\lgc{K(A)}$. The crucial ingredient of our proof will be the fact that an $\lgc{LK'(A)}$-tableau for an $\lang_{\lgc{A_m}}^\bo$-formula  always consists of just one branch, and hence a single inconsistent system of linear inequations can be associated with each valid $\lang_{\lgc{A_m}}^\bo$-formula. 

We begin by proving two lemmas for $\lgc{K(A)}$-valid sequents of a certain form, recalling that sequents contain only $\lang_{\lgc{A_m}}^\bo$-formulas by definition.

\begin{lem}\label{l:separation}
Let $ \Ga, \Pi \seq \Si, \De$ be a $\lgc{K(A)}$-valid sequent such that no variable occurs in both $\Ga \uplus \De$ and $\Pi \uplus \Si$. Then $\Ga \seq \De$  and $\Pi \seq \Si$ are both $\lgc{K(A)}$-valid.
\end{lem}
\begin{proof}
Suppose contrapositively that $\not \mdl{\lgc{K(A)}} \Ga \seq \De$. Then there exists a $\lgc{K(A)}$-model $\fram{M} = \langle W, R, V \rangle$ and $x \in W$ such that $V(\mathcal{I}(\Ga \seq \De),x) < 0$. Since $\Ga \uplus \De$ and $\Pi \uplus \Si$ have disjoint sets of variables, we may assume without loss of generality that $V(p,y) = 0$ for all $p \in \var$ occurring in $\Pi \uplus \Si$ and $y \in W$. A simple induction yields also that $V(\f,y) = 0$ for all $\f \in \Pi \uplus \Si$ and $y \in W$. But then $V(\mathcal{I}(\Ga, \Pi \seq \Si,\De),x) < 0$. So  $\not \mdl{\lgc{K(A)}} \Ga, \Pi \seq \Si, \De$. The case where $\not \mdl{\lgc{K(A)}} \Pi \seq \Si$ follows by symmetry. 
\end{proof}

\begin{lem}\label{l:mix}
Let $\bo \Ga, \Pi \seq \Si, \bo \De$ be a $\lgc{K(A)}$-valid sequent such that $\Pi$ and $\Si$ contain only variables. Then $\Pi = \Si$ and $\bo \Ga \seq \bo \De$ is $\lgc{K(A)}$-valid.
\end{lem}
\begin{proof}
Suppose that $\mdl{\lgc{K(A)}} \bo \Ga, \Pi \seq \Si, \bo \De$. It suffices to show that $\Pi = \Si$, since then clearly also $\mdl{\lgc{K(A)}} \bo \Ga \seq \bo \De$. Suppose for a contradiction that $\Pi \not = \Si$. Without loss of generality, some $p \in \var$ occurs strictly more times in $\Pi$ than $\Si$. Consider a $\lgc{K(A)}$-model $\fram{M} = \langle \{x\}, \emptyset, V \rangle$ with one irreflexive world $x$ satisfying $V(p,x) = 1$ and $V(q,x)  = 0$ for all $q \in \var \setminus \{p\}$. Then $V(\mathcal{I}( \bo \Ga, \Pi \seq \Si, \bo \De),x) < 0$ and so $\not \mdl{\lgc{K(A)}} \bo \Ga, \Pi \seq \Si, \bo \De$, a contradiction. 
\end{proof}

To deal with $\lgc{K(A)}$-valid sequents in general, we use the fact that for such a sequent, there must exist a corresponding closed complete $\lgc{LK'(A)}$-tableau with one branch and an associated inconsistent set of inequations. We use this set of inequations to show that the rule $\boxknr{k,m}$ for suitable $k,m$ can be applied backwards to the sequent to obtain $\lgc{K(A)}$-valid sequents containing formulas of strictly smaller modal depth. To this end, it will be helpful to extend some of the notions for the labelled tableau calculus $\lgc{LK'(A)}$ to sequents. We define a {\em complete $\lgc{LK'(A)}$-tableau} for a sequent $\Ga \seq \De$  to be a tableau beginning with the active inequation $\lab{\Ga}{1} > \lab{\De}{1}$ and relation $r12$, constructed according to steps (2)--(7). Consulting the proof of Theorem~\ref{t:LabelledEquivSemantics}, we obtain the following result.

\begin{cor}\label{c:sequentlabelledcalculus}
There exists a  closed complete $\lgc{LK'(A)}$-tableau for a sequent $\Ga \seq \De$ if and only if  $\Ga \seq \De$ is $\lgc{K(A)}$-valid.
\end{cor}

To argue about the inconsistency of a system of inequations associated to a tableau, we recall some basic notions from linear programming. Let $S$ be a system of  inequations of the form $I_i = (f_i(\bar{x}) > g_i(\bar{x}))$ ($i \in \{1,\ldots,n\}$) and $J_j = (h_j(\bar{x}) \ge k_j(\bar{x}))$  ($j \in \{1,\ldots,m\}$) where each $f_i,g_i,h_j,k_j$ is a positive linear sum of variables in $\bar{x}$. Then $S$ is inconsistent over $\R$ if and only if there exists an inequation given by a linear combination of these inequations
\[
\begin{array}{rcl}
L_S & = &  \sum_{i = 1}^n \lambda_i I_i + \sum_{j = 1}^m \mu_j J_j
\end{array}
\]
where $\lambda_1,\ldots,\lambda_n \in \N$ (not all zero) and $\mu_1,\ldots,\mu_m \in \N$ such that
\[
\begin{array}{rcl}
\lambda_1 f_1 + \ldots + \lambda_n f_n + \mu_1 h_1 + \ldots + \mu_m h_m & = & \lambda_1 g_1 + \ldots + \lambda_n g_n + \mu_1 k_1 + \ldots + \mu_m k_m.
\end{array}
\]
We say that $L_S$ is {\em inconsistent} and that each inequation $f_i(\bar{x}) > g_i(\bar{x})$ or $h_j(\bar{x}) \ge k_j(\bar{x})$ is {\em used} $\lambda_i$ or $\mu_j$ times, respectively, in $L_S$.

Given a labelled inequation $I = \lab{\Ga_1}{\vecn{k_1}} \aineq \lab{\De_1}{\vecn{l_1}}$, let $I^E =  \lab{\Ga_2}{\vecn{k_2}} \aineq \lab{\De_2}{\vecn{l_2}}$ be the inequation obtained by applying the rules for $\to$ in $\lgc{LK'(A)}$ to $I$ exhaustively. By further replacing each boxed formula $\bo \f$ with $\newv{\bo \f}$, we obtain the \emph{reduced form} $I^R$ of $I$, saying that $I$ is \emph{in reduced form} if $I = I^R$. We now have all the required tools to prove our main lemma.

\begin{lem}\label{l:main}
Let $\bo \Ga \seq \bo \p_1,\ldots, \bo \p_m$ be a $\lgc{K(A)}$-valid sequent. Then there exist $k \in \N \setminus\!\{0\}$ and multisets of $\lang_{\lgc{A_m}}^\bo$-formulas $\Ga_0,\Ga_1,\ldots,\Ga_m$ such that 
\begin{enumerate}[label=\rm (\roman*)]
\item	$k\Ga = \Ga_0 \uplus \Ga_1 \uplus \ldots \uplus \Ga_m$
\item	$\Ga_0 \seq$ \, and\, $\Ga_i \seq k[\p_i]$ for $i \in \{1,\ldots,m\}$ are all $\lgc{K(A)}$-valid.
\end{enumerate}
\end{lem}
\begin{proof}
Let $\Ga = [\f_1,\ldots,\f_n]$. By assumption, $\mdl{\lgc{K(A)}} \bo \f_1,\ldots,\bo \f_n \seq \bo \p_1,\ldots,\bo \p_m$, and, by Corollary~\ref{c:sequentlabelledcalculus}, we obtain a complete closed tableau $T$ in $\lgc{LK'(A)}$ that begins with 
\[
\lab{{\bo \f_1}}{1},\ldots,\lab{{\bo \f_n}}{1} > \lab{{\bo \p_1}}{1},\ldots, \lab{{\bo \p_m}}{1} \quad \text{and} \quad r12.
\]
This tableau will contain the inequation
\[
\begin{array}{rcl}
I & = & \lab{\newv{\bo \f_1}}{1},\ldots,\lab{\newv{\bo \f_n}}{1} > \lab{\newv{\bo \p_1}}{1},\ldots, \lab{\newv{\bo \p_m}}{1}
\end{array}
\]
and for new labels $y_1,\ldots,y_m \in \N$, the inequations
\[
\begin{array}{rclcrcl}
I_1 & = & \lab{\newv{\bo \p_1}}{1} \ge \lab{\p_1}{y_1} & \quad \ldots \quad & I_m & = &  \lab{\newv{\bo \p_m}}{1} \ge \lab{\p_m}{y_m}.
\end{array}
\]
Let us fix $y_0 = 2$. Then $T$ contains for each $i \in \{1, \ldots, n\}$ and  $j \in \{ 0, \ldots, m\}$, an inequation
\[
\begin{array}{rcl}
I_{i,j} & = & \lab{\f_i}{y_j} \ge \lab{\newv{\bo \f_i}}{1}.
\end{array}
\]
Consider now the set of inequations associated to $T$
\[
\begin{array}{rcl}
S &  = & \{I\} \cup \{I_j^R : 1 \le j \le m\} \cup  \{I_{ij}^R : 1 \le i \le n, \, 0 \le j \le m\} \cup S',
\end{array}
\]
noting that the inequations in $S'$ are obtained by applying rules of $\lgc{LK'(A)}$ to inequations in $\{I_j^E : 1 \le j \le m\} \cup  \{I_{ij}^E : 1 \le i \le n, \, 0 \le j \le m\}$. Since $T$ is closed, $S$ is inconsistent over $\R$. Hence there is an inconsistent linear combination $L_S$ of the inequations in $S$. The following observations can be confirmed  by simple inductions on the height of $T$:

\begin{enumerate}[label=\rm (\roman*)]

\item The (reduced form) inequation $I$ is the only strict inequation occurring in $S$, and hence must be used $k$ times in $L_S$ for some $k \in \N \setminus\!\{0\}$. \smallskip

\item For each $j \in \{1, \ldots, m\}$, $\lab{\newv{\bo \p_j}}{1}$ occurs in $S$ only in $I$ and in the reduced form $I_j^R$ of $I_j$; hence, by (i), $I_j^R$ must also be used $k$ times in $L_S$. \smallskip

\item For each $i \in \{1, \ldots,n\}$, $\lab{\newv{\bo \f_i}}{1}$ occurs in $S$ only in $I$ and in the reduced forms $I_{i,j}^R$ of $I_{i,j}$ for $j \in \{0, \ldots, m\}$; hence, given that $I_{i,j}^R$ is used in the linear combination $\lambda_{i,j}$ times, we obtain $\lambda_{i,0} + \lambda_{i,1} + \ldots + \lambda_{i,m} = k$; in particular, not all $\lambda_{i,j}$ are zero. \smallskip

\end{enumerate}

\noindent
The inconsistent linear combination of the inequations in $S$ is therefore
\[
\begin{array}{rcl}
L_S & = & k I + \sum_{j = 1}^m k I^R_j + \sum_{i = 1}^n  \sum_{j = 0}^m \lambda_{i,j} I^R_{i,j} + L_{S'}.
\end{array}
\]
We define multisets of formulas
\[
\begin{array}{rcl}
\Ga_j & = & \lambda_{1,j} [\f_1], \ldots, \lambda_{n,j} [\f_n] \ \text{ for $j \in \{0,\ldots,m\}$}\\[.05in]
\De  & = & k[\newv{\bo \f_1}],\ldots, k[\newv{\bo \f_n}], k[\newv{\bo \p_1}],\ldots, k[\newv{\bo \p_m}].
\end{array}
\] 
Note that, as required, $k\Ga = \Ga_0 \uplus \Ga_1 \uplus \ldots \uplus \Ga_m$. Consider now the inequation
\[
\begin{array}{rcl}
 J	& = & kI + \sum_{j=1}^{m} kI_j + \sum_{i=1}^{n}  \sum_{j=0}^{m} \lambda_{i,j} I_{i,j}\\[.05in]
	& =  & \lab{\Ga_0}{y_0}, \lab{\Ga_1}{y_1}, \ldots, \lab{\Ga_m}{y_m}, \lab{\De}{1} > \lab{\De}{1}, k[\lab{\p_1}{y_1}], \ldots,  k[\lab{\p_m}{y_m}].
\end{array}
\]
Then $L_S =  J^R + L_{S'}$ and the set of inequations $S^* = \{J^R\} \cup S'$ is inconsistent over $\R$.

Recall that each (reduced form) inequation in $S'$ is obtained by applying rules of $\lgc{LK'(A)}$ to the inequations $\{I_j^E : 1 \le j \le m\} \cup  \{I_{ij}^E : 1 \le i \le n, \, 0 \le j \le m\}$. But following the procedure for building a complete $\lgc{LK'(A)}$-tableau, the inequations in $S'$ are obtained by first applying the rules $\bolrn$ and $\borrn$. Hence these inequations in $S'$ and $J^R$ are also obtained by first applying the rules $\bolrn$ and $\borrn$ to $J^E$ and then continuing as before. 

Now for each  $j \in \{0,\ldots,m\}$, let $\var_j \subseteq \var$ be a countably infinite set such that $\var_0 \cap \var_1 \cap \ldots \cap \var_m = \emptyset$, and let $h_j \colon \var \to \var_j$ be a bijective map that extends in the obvious way to all formulas and multisets of formulas. Consider the inequation
\[ 
\begin{array}{rcl}
J' & = & \lab{h_0(\Ga_0)}{1}, \lab{h_1(\Ga_1)}{1}, \ldots, \lab{h_m(\Ga_m)}{1}, \lab{\De}{1} > \lab{\De}{1}, k[\lab{h_1(\p_1)}{1}], \ldots,  k[\lab{h_m(\p_m)}{1}].
\end{array}
\]
An easy induction on the height of a tableau shows that applying the rules of $\lgc{LK'(A)}$ to $J'$ and relation $r12$ also produces  a set of inequations that is inconsistent over $\R$. But then by Corollary~\ref{c:sequentlabelledcalculus}, 
\[
\mdl{\lgc{K(A)}} h_0(\Ga_0), h_1(\Ga_1), \ldots, h_m(\Ga_m), \De \seq \De, k[h_1(\p_1)],\ldots, k[h_m(\p_m)].
\]
Applying Lemma~\ref{l:separation} repeatedly, we obtain 
\[
\mdl{\lgc{K(A)}} h_0(\Ga_0) \seq \quad \text{ and } \quad \mdl{\lgc{K(A)}} h_i(\Ga_i) \seq k[h_i(\p_i)] \, \text{ for }i \in \{1,\ldots,m\},
\]
and hence, renaming variables,
\[
\mdl{\lgc{K(A)}} \Ga_0 \seq \quad \text{ and } \quad \mdl{\lgc{K(A)}} \Ga_i \seq k[\p_i] \, \text{ for }i \in \{1,\ldots,m\}
\]
as required.
\end{proof}

\begin{prop}\label{p:sequentcalculuscompleteness}
Let $\Ga \seq \De$ be a $\lgc{K(A)}$-valid sequent. Then $\der{\lgc{GK(A_m)}} \Ga \seq \De$. 
\end{prop}

\begin{proof}
We prove the claim by induction on the lexicographically ordered pair consisting of the modal depth of $\mathcal{I}(\Ga \seq \De)$ and the sum of the complexities of the formulas in $\Ga \uplus \De$. 

For the base case, suppose that $\mdl{\lgc{K(A)}} \Ga \seq \De$ and that both $\Ga$ and $\De$ contain only variables. Then, by Lemma~\ref{l:mix}, we obtain $\Ga = \De$. Hence, by $\idr$, we get $\der{\lgc{GK(A_m)}} \Ga \seq \De$.

For the inductive step, suppose first that $\mdl{\lgc{K(A)}} \Ga, \f \to \p \seq \De$. Then also $\mdl{\lgc{K(A)}} \Ga, \p \seq \f, \De$. So by the induction hypothesis, $\der{\lgc{GK(A_m)}} \Ga, \p \seq \f, \De$. Hence, by $\ilr$, we get $\der{\lgc{GK(A_m)}} \Ga, \f \to \p \seq \De$. The case where $\f \to \p$ occurs on the right is very similar.

Now suppose that $\mdl{\lgc{K(A)}}\bo \Ga, \Pi \seq \Si, \bo \p_1,\ldots,\bo \p_m$ where $\Pi$ and $\Si$ contain only variables.  By Lemma~\ref{l:mix}, we obtain $\Pi= \Si$ and $\mdl{\lgc{K(A)}} \bo \Ga \seq  \bo \p_1,\ldots,\bo \p_m$. By $\idr$, we get $\der{\lgc{GK(A_m)}}\Pi \seq  \Si$.  Moreover, by Lemma~\ref{l:main}, there exist $k \in \N \setminus\!\{0\}$ and multisets of $\lang_{\lgc{A_m}}^\bo$-formulas $\Ga_0,\Ga_1,\ldots,\Ga_m$ such that 
\begin{enumerate}[label=\rm (\roman*)]
\item	$k\Ga = \Ga_0 \uplus \Ga_1 \uplus \ldots \uplus \Ga_m$
\item	$\mdl{\lgc{K(A)}} \Ga_0 \seq$ \, and\, $\mdl{\lgc{K(A)}} \Ga_i \seq k[\p_i]$ for $i \in \{1,\ldots,m\}$.
\newcounter{tempSaveCounter}
\setcounter{tempSaveCounter}{\value{enumi}}
\end{enumerate}
But then by the induction hypothesis also
\begin{enumerate}[label=\rm (\roman*)]
\setcounter{enumi}{\value{tempSaveCounter}}
\item	$\der{\lgc{GK(A_m)}} \Ga_0 \seq$ \, and\, $\der{\lgc{GK(A_m)}} \Ga_i \seq k[\p_i]$ for $i \in \{1,\ldots,m\}$.
\end{enumerate}
Hence, using the $\lgc{GK(A_m)}$-derivable rule $\boxknr{k,m}$, we obtain $\der{\lgc{GK(A_m)}} \bo \Ga \seq \bo \p_1,\ldots,\bo \p_m$. Finally, using $\mixr$, we obtain $\der{\lgc{GK(A_m)}} \bo \Ga, \Pi \seq \Si, \bo \p_1,\ldots,\bo \p_m$ as required.
\end{proof}

Our main theorem now follows as a direct combination of Propositions~\ref{p:axiomsystemsound},~\ref{p:sequentcalculusaxiomsystemequivalent}, and~\ref{p:sequentcalculuscompleteness}.

\begin{thm}\label{t:main}
The following are equivalent for any $\f \in \fm(\lang_{\lgc{A_m}}^\bo)$:
\begin{enumerate}[label=\rm (\arabic*)]
\item $\mdl{\lgc{K(A)}} \f$.
\item	$\der{K(A_m)} \f$.
\item	$\der{\lgc{GK(A_m)}} \, \seq\! \f$.
\end{enumerate}
\end{thm}

Let us remark finally that, since any  $\lgc{LK'(A)}$-tableau for an $\lang_{\lgc{A_m}}^\bo$-formula has just one branch, we obtain (consulting the proof of Theorem~\ref{t:complexity}) a smaller upper bound for the complexity of checking $\lgc{K(A)}$-validity in this fragment.

\begin{thm}
The problem of checking if $\f \in \fm(\lang_\lgc{A_m}^\bo)$ is $\lgc{K(A)}$-valid is in {\sc EXPTIME}.
\end{thm}


\section{Concluding Remarks}

This paper takes a  significant step towards a proof-theoretic account of continuous modal logics: many-valued modal logics with connectives interpreted locally by continuous functions over sets of real numbers. We have introduced here a minimal modal extension $\lgc{K(A)}$ of Abelian logic (see~\cite{mey:ab,cas:ab,met:seq}), where propositional connectives are interpreted using lattice-ordered group operations over the real numbers, and shown that the modal \L ukasiewicz logic $\lgc{K(\mathrmL)}$ studied in~\cite{HT13} is a fragment of this logic with an additional constant. We have provided a labelled tableau calculus for $\lgc{K(A)}$ and established a {\sc coNEXPTIME} upper bound for checking validity. More significantly, for the modal-multiplicative fragment of $\lgc{K(A)}$, we have obtained both a sequent calculus that admits cut-elimination and an axiomatization without infinitary rules. Notably, this latter result was established using the completeness of the labelled tableau calculus to derive a corresponding proof in the sequent calculus. The more standard algebraic approach to proving completeness of many-valued modal logics, employed, e.g., for finite-valued {\L}ukasiewicz modal logics in~\cite{HT13}, proceeds by constructing a canonical model as the set of maximal filters of the Lindenbaum-Tarski algebra of the logic. For finite-valued {\L}ukasiewicz modal logics, completeness is proved using the fact that the appropriate reduct of this algebra is semi-simple, which is not applicable in the infinite-valued case or for the modal-multiplicative fragment of $\lgc{K(A)}$.

Clearly, there are many open questions still to be addressed. The most pressing issue is to find an axiomatization and algebraic semantics for the full logic $\lgc{K(A)}$. We conjecture that such an axiomatization can be obtained by extending the axiom system $\lgc{HA}$ for Abelian logic with the axiom schema (K), (D$_n$) ($n \ge 2$) and rules (mp), (nec) from Figure~\ref{f:kz}, and the axiom schema $(\bo \f \land \bo \p) \to \bo (\f \land \p)$. It can be shown using methods of abstract algebraic logic that this axiom system is sound and complete with respect to a corresponding variety of algebras with a lattice-ordered abelian group reduct; the difficulty of course is to prove that the axiomatization is complete with respect to the frame semantics of $\lgc{K(A)}$, perhaps by extending the proof for the modal-multiplicative fragment (using the labelled tableau calculus and a Gentzen-style calculus), or via an alternative representation of the algebras. Such a proof would provide the basis for an axiomatization and algebraic semantics for $\lgc{K(\mathrmL)}$, and, more generally, a starting point for a J{\'o}nsson-Tarski-style account of the relationship between relational and algebraic semantics for these logics. Note that we can already  develop such a relationship for the modal-multiplicative fragment axiomatized in this paper, but the algebras corresponding to the axiom system $\lgc{K(A_m)}$ will not form a variety.

We have focussed in this work only on the minimal modal extension of Abelian logic. However, adapting the Kripke semantics and labelled tableau calculi to other  (e.g., reflexive, symmetric, transitive) classes of frames is a straightforward exercise. More challenging is the problem of adapting the completeness proofs for the modal-multiplicative fragment to suitably extended axiom systems and sequent calculi. For the reflexive case, completeness proofs, similar to those given here, can be obtained for the extension of the  axiom system $\lgc{K(A_m)}$ with the axiom schema $\bo \f \to \f$ and the sequent calculus $\lgc{GK(A_m)}$ with the rule
\[
\infer{\Ga, \bo \f \seq \De}{\Ga, \f \seq \De}
\]
However, a general approach for tackling different classes of frames is still lacking.

Finally, it remains to determine whether the upper bounds given here for the complexity of checking $\lgc{K(A)}$-validity are optimal. Let us just note that it makes sense to first investigate the {\sc EXPTIME} upper bound for the modal-multiplicative fragment, before considering the {\sc coNEXPTIME} upper bound for the full logic $\lgc{K(A)}$ and indeed also $\lgc{K(\mathrmL)}$.

\bibliographystyle{plain}

\begin{thebibliography}{10}

\bibitem{BBP17}
F.~Baader, S.~Borgwardt, and R.~Pe{\~{n}}aloza.
\newblock Decidability and complexity of fuzzy description logics.
\newblock {\em {KI}}, 31(1):85--90, 2017.

\bibitem{BM08}
M.~Baaz and G.~Metcalfe.
\newblock Herbrand's theorem, skolemization and proof systems for first-order {\L}ukasiewicz logic.
\newblock {\em Journal of Logic and Computation}, 20(1):35--54, 2008.

\bibitem{Bar17}
S.~Baratella.
\newblock Continuous propositional modal logic.
\newblock Submitted (available at
  \url{http://www.science.unitn.it/~baratell/CPropML.pdf)}.

\bibitem{BCE11}
F.~Bou, M.~Cerami, and F.~Esteva.
\newblock Finite-valued {{\L}}ukasiewicz modal logic is {PSPACE}-complete.
\newblock In {\em Proceedings of {IJCAI} 2011}, pages 774--779, 2011.

\bibitem{BEGR11}
F.~Bou, F.~Esteva, L.~Godo, and R.~Rodr{\'\i}guez.
\newblock On the minimum many-valued logic over a finite residuated lattice.
\newblock {\em Journal of Logic and Computation}, 21(5):739--790, 2011.

\bibitem{CMRR17}
X.~Caicedo, G.~Metcalfe, R.~Rodr{\'\i}guez, and J.~Rogger.
\newblock Decidability in order-based modal logics.
\newblock {\em Journal of Computer System Sciences}, 88:53--74, 2017.

\bibitem{CR10}
X.~Caicedo and R.~Rodr{\'\i}guez.
\newblock Standard {G}{\"o}del modal logics.
\newblock {\em Studia Logica}, 94(2):189--214, 2010.

\bibitem{CR12}
X.~Caicedo and R.~Rodr{\'\i}guez.
\newblock Bi-modal {G}{\"o}del logic over [0,1]-valued {K}ripke frames.
\newblock {\em Journal of Logic and Computation}, 25(1):37--55, 2015.

\bibitem{cas:ab}
E.~Casari.
\newblock Comparative logics and abelian $\ell$-groups.
\newblock In C.~Bonotto, R.~Ferro, S.~Valentini, and A.~Zanardo, editors, {\em
  Logic Colloquium '88}, pages 161--190. Elsevier, 1989.

\bibitem{DeMa79}
N.~Dershowitz and Z.~Manna.
\newblock Proving termination with multiset orderings.
\newblock {\em Communications of the Association for Computing Machinery},
  22:465--476, 1979.

\bibitem{DG07}
D.~Diaconescu and G.~Georgescu.
\newblock Tense operators on {MV}-algebras and {{\L}}ukasiewicz-{M}oisil
  algebras.
\newblock {\em Fundamenta Informaticae}, 81(4):379--408, 2007.

\bibitem{DMS16}
D.~Diaconescu, G.~Metcalfe, and L.~Schn{\"u}riger.
\newblock Axiomatizing a real-valued modal logic.
\newblock In {\em Proceedings of {AiML} 2016}, pages 236--251, 2016.

\bibitem{Fitting91}
M.~C. Fitting.
\newblock Many-valued modal logics.
\newblock {\em Fundamenta Informaticae}, 15(3--4):235--254, 1991.

\bibitem{Fitting92}
M.~C. Fitting.
\newblock Many-valued modal logics {II}.
\newblock {\em Fundamenta Informaticae}, 17:55--73, 1992.

\bibitem{GHE03}
L.~Godo, P.~H{\'a}jek, and F.~Esteva.
\newblock A fuzzy modal logic for belief functions.
\newblock {\em Fundamenta Informaticae}, 57(2--4):127--146, 2003.

\bibitem{GR99}
L.~Godo and R.~Rodr{\'\i}guez.
\newblock A fuzzy modal logic for similarity reasoning.
\newblock In {\em Fuzzy Logic and Soft Computing}, pages 33--48. Kluwer, 1999.

\bibitem{Haj98}
P.~H{\'a}jek.
\newblock {\em Metamathematics of Fuzzy Logic}.
\newblock Kluwer, Dordrecht, 1998.

\bibitem{Haj05}
P.~H{\'a}jek.
\newblock Making fuzzy description logic more general.
\newblock {\em Fuzzy Sets and Systems}, 154(1):1--15, 2005.

\bibitem{HEGG94}
P.~H{\'a}jek, D.~Harmancov{\'a}, F.~Esteva, P.~Garcia, and L.~Godo.
\newblock On modal logics for qualitative possibility in a fuzzy setting.
\newblock In {\em Proceedings of {UAI} 1994}, pages 278--285, 1994.

\bibitem{HHV95}
P.~H{\'a}jek, D.~Harmancov{\'a}, and R.~Verbrugge.
\newblock A qualitative fuzzy possibilistic logic.
\newblock {\em International Journal of Approximate Reasoning}, 12:1--19, 1995.

\bibitem{HT13}
G.~Hansoul and B.~Teheux.
\newblock Extending {{\L}}ukasiewicz logics with a modality: Algebraic approach
  to relational semantics.
\newblock {\em Studia Logica}, 101(3):505--545, 2013.

\bibitem{Khachiyan79}
L.~G. Khachiyan.
\newblock A polynomial algorithm in linear programming.
\newblock {\em Soviet Mathematics Doklady}, 20:191--194, 1979.

\bibitem{KPS13}
A.~Kulacka, D.~Pattinson, and L.~Schr{\"{o}}der.
\newblock Syntactic labelled tableaux for {{\L}}ukasiewicz fuzzy {ALC}.
\newblock In {\em Proceedings of {IJCAI} 2013}, pages 962--968, 2013.

\bibitem{KP10}
C.~Kupke and D.~Pattinson.
\newblock On modal logics of linear inequalities.
\newblock In {\em Proceedings of AiML 2010}, pages 235--255. King's College
  Publications, 2010.

\bibitem{MM14}
M.~Marti and G.~Metcalfe.
\newblock Hennessy-{M}ilner properties for many-valued modal logics.
\newblock In {\em Proceedings of AiML 2014}, pages 407--420. King's College
  Publications, 2014.
  
\bibitem{MM17}
M.~Marti and G.~Metcalfe.
\newblock Expressivity of chain-based modal logics.
\newblock {\em Archive for Mathematical Logic}. To appear.

\bibitem{MS13}
M.~Mio and A.~Simpson.
\newblock {\L}ukasiewicz mu-calculus.
\newblock In {\em Proceedings Workshop on Fixed Points in Computer Science},
  volume 126 of {\em EPCTS}, pages 87--104. Open Publishing Association, 2013.

\bibitem{MO11}
G.~Metcalfe and N.~Olivetti.
\newblock Towards a proof theory of {{G}}{\"o}del modal logics.
\newblock {\em Logical Methods in Computer Science}, 7(2):1--27, 2011.

\bibitem{met:seq}
G.~Metcalfe, N.~Olivetti, and D.~Gabbay.
\newblock Sequent and hypersequent calculi for abelian and {{\L}}ukasiewicz
  logics.
\newblock {\em ACM Transactions on Computational Logic}, 6(3):578--613, 2005.

\bibitem{MOG08}
G.~Metcalfe, N.~Olivetti, and D.~Gabbay.
\newblock {\em Proof Theory for Fuzzy Logics}, volume~36 of {\em Applied
  Logic}.
\newblock Springer, 2008.

\bibitem{mey:ab}
R.~K. Meyer and J.~K. Slaney.
\newblock Abelian logic from {A} to {Z}.
\newblock In {\em Paraconsistent Logic: Essays on the Inconsistent}, pages
  245--288. Philosophia Verlag, 1989.

\bibitem{priest:many}
G.~Priest.
\newblock Many-valued modal logics: a simple approach.
\newblock {\em Review of Symbolic Logic}, 1:190--203, 2008.

\bibitem{SDK09}
S.~Schockaert, M.~De Cock, and E.~Kerre.
\newblock Spatial reasoning in a fuzzy region connection calculus.
\newblock {\em Artificial Intelligence}, 173(2):258--298, 2009.

\bibitem{straccia01a}
U.~Straccia.
\newblock Reasoning within fuzzy description logics.
\newblock {\em Journal of Artificial Intelligence Research}, 14:137--166, 2001.

\end{thebibliography}

\end{document}